\documentclass[lettersize,journal]{IEEEtran}
\usepackage{amsmath,amsfonts}
\usepackage{algorithmic}
\usepackage{algorithm}
\usepackage{array}
\usepackage[caption=false,font=normalsize,labelfont=sf,textfont=sf]{subfig}
\usepackage{textcomp}
\usepackage{stfloats}
\usepackage{url}
\usepackage{verbatim}
\usepackage{graphicx}
\usepackage{cite}
\usepackage{siunitx}
\usepackage{bm}
\usepackage{booktabs}
\usepackage{amsthm}
\usepackage{amssymb}
\theoremstyle{plain}
\usepackage{bbold}
\newtheorem{assumption}{Assumption}
\newtheorem{lemma}{Lemma}
\newtheorem{thm}{Theorem}
\newtheorem{defi}{Definition}

\newtheorem{coro}{Corollary}

\newcommand{\algrule}[1][.2pt]{\par\vskip.5\baselineskip\hrule height #1\par\vskip.5\baselineskip}
\begin{document}

\title{Over-the-Air Federated Learning in MIMO Cloud-RAN Systems}

\author{Haoming~Ma,
Xiaojun~Yuan,~\IEEEmembership{Senior Member,~IEEE,}
Zhi~Ding,~\IEEEmembership{Fellow,~IEEE,}
	
\thanks{H. Ma and X. Yuan are with the National Key Laboratory of Science and Technology on Communication, the University of Electronic Science and Technology of China, Chengdu, China (e-mail:hmma@std.uestc.edu.cn; xjyuan@uestc.edu.cn). Z. Ding is with the Department of Electrical and Computer Engineering, University of California at Davis, Davis, CA 95616 USA (e-mail:zding@ucdavis.edu). The corresponding author is Xiaojun Yuan.
}
}


\IEEEpubid{0000--0000/00\$00.00~\copyright~2021 IEEE}

\maketitle

\begin{abstract}
To address the limitations of traditional over-the-air federated learning (OA-FL) such as limited server coverage and low resource utilization, we propose an OA-FL in MIMO cloud radio access network (MIMO Cloud-RAN) framework, where edge devices upload (or download) model parameters to the cloud server (CS) through access points (APs). Specifically, in every training round, there are three stages: edge aggregation; global aggregation; and model updating and broadcasting. To better utilize the correlation 
among APs, called inter-AP correlation, we propose modeling the global aggregation stage as a lossy distributed source coding (L-DSC) problem to make analysis from the perspective of rate-distortion theory. We further analyze the performance of the proposed OA-FL in MIMO Cloud-RAN framework. Based on the analysis, we formulate a communication-learning optimization problem to improve the system performance by considering the inter-AP correlation. To solve this problem, we develop an algorithm by using alternating optimization (AO) and majorization-minimization (MM), which effectively improves the FL learning performance. Furthermore, we propose a practical design that demonstrates the utilization of inter-AP correlation. The numerical results show that the proposed practical design effectively leverages inter-AP correlation, and outperforms other baseline schemes.
\end{abstract}


\section{Introduction}
With the growing amount of data stored on mobile edge devices, there is an increasing interest in providing artificial intelligence (AI) services, such as computer vision \cite{he_deep_2016} and natural language processing \cite{young_recent_2018}, at the wireless edge. Traditional machine learning (ML) methods require uploading local data to a central node for model training, which unfortunately incurs significant communication costs and raises concerns about data privacy. To address these issues, federated learning (FL) has emerged as a promising framework for distributed and confidential model training \cite{goetz_active_2019}. In the FL framework, each edge device trains its local model using its local data and sends its model update to a cloud server (CS) (referred to as uplink transmission). The CS aggregates the local model updates and updates the global model parameters, and then broadcasts the updated global model to the edge devices (referred to as downlink transmission). Compared to centralized learning, FL significantly reduces the communication burden and the risk of data breaches, and therefore becomes a compelling option for machine learning applications at the wireless edge.

The FL uplink involves transmitting a massive amount of model updates from distributed edge devices to the CS, which poses a critical communication bottleneck due to limited uplink channel resources, such as bandwidth, time, and space\cite{kairouz_advances_2021}. Over-the-air computation has been adopted to improve the communication efficiency of the FL uplink by supporting analog model uploading from massive edge devices \cite{nazer_computation_2007}. Unlike orthogonal resource allocation among devices to avoid interference, over-the-air FL (OA-FL) enables devices to share radio resources during model uploading by leveraging the signal-superposition property of analog transmission for model aggregation over the air. Pioneering work has confirmed that OA-FL has strong noise tolerance \cite{zhu_one-bit_2021} and can reduce latency substantially compared with the FL schemes based on conventional orthogonal multiple access (OMA) protocols. The authors in \cite{wen_reduced-dimension_2019} have further extended over-the-air FL to the multiple input multiple output (MIMO) channel, and leveraged the beamforming technology to enhance the communication quality.

Despite the benefits of OA-FL mentioned above, the limited service coverage of a single CS makes it challenging to access more devices \cite{chen_china_2013}. In general, expanding service coverage by simply replicating CS inevitably results in significant energy consumption, construction investment, site support, leasing, and maintenance costs, making it a costly solution \cite{richter_energy_2009}. Additionally, due to user mobility, traffic on each CS experiences fluctuations, commonly known as the "tide effect." \cite{chen_china_2013}. These CS must be designed to handle maximum traffic rather than average traffic, resulting in a significant waste of processing resources and power during idle periods. 

To address these issues, an effective solution is the MIMO Cloud-RAN architecture. This architecture consists of multiple access points (APs), also known as remote radio heads (RRHs), with each AP serving a specific set of mobile devices. These APs transmit (or receive) signals to (or from) mobile devices through a MIMO-based radio access network and upload (or download) data payloads to (or from) the CS via a fronthaul network \cite{park_fronthaul_2014}. MIMO Cloud-RAN provides flexible network planning and deployment \cite{pompili_dynamic_2015}, resulting in a significant increase in system coverage at a low cost. Furthermore, it enables the full utilization of idle processing resources and reduces the impact of the tide effect by centralizing computing resources into a resource pool of the CS.

In this paper, we introduce the OA-FL into MIMO Cloud-RAN and proposed a OA-FL in MIMO Cloud-RAN framework. More specifically, our proposed framework consists of three stages: edge aggregation, in which each AP aggregates local updates from edge devices through the MIMO-based radio access network \cite{wen_reduced-dimension_2019} over the air to generate edge updates; global aggregation, in which the CS aggregates edge updates from APs through the fronthaul network to generate a global update; and model updating and broadcasting, in which the CS updates global model parameters and broadcasts it to every AP through the fronthaul network, which then broadcasts it to their served devices.

Additionaly, we observed that the local updates in FL are typically correlated \cite{zhong_over--air_2021}, resulting in correlated edge updates. This correlation, referred to as inter-AP correlation, can be leveraged to significantly reduce the communication cost of global aggregation. To better utilize the inter-AP correlation, we propose modeling the global aggregation stage as a lossy distributed source coding (L-DSC) problem to make analysis from the perspective of rate-distortion theory. Based on this analysis, we further analyze the performance of the proposed OA-FL in MIMO Cloud-RAN framework. Furthermore, we formulate a communication-learning optimization problem to improve the system performance by considering the inter-AP correlation. To solve this problem, we develop an algorithm by using alternating optimization (AO) and majorization-minimization (MM), which effectively improves the FL learning performance. 

Subsequently, we propose a practical design that demonstrates the utilization of inter-AP correlation. Our design includes an encoding function and a joint decoding function, which enable efficient compression and reconstruction of the model updates. Specifically, the encoding function consists of a compression module, quantization module, and error accumulation module, which work together to compress the edge update, add an error accumulation term, and quantize the resulting vector to satisfy the rate constraint. On the other hand, the decoding function utilizes a neural network structure to leverage inter-AP correlation for estimating the global update. The numerical results show that the proposed practical design effectively leverages inter-AP correlation, and outperforms other baseline schemes.

The remainder of this paper is structured as follows. In Section \ref{Preliminaries}, we provide an overview of the FL process and MIMO Cloud-RAN. In Section \ref{c-ran}, we present the OA-FL in MIMO Cloud-RAN framework. In Section \ref{analysis}, we make analysis for the L-DSC problem and the proposed OA-FL in MIMO Cloud-RAN framework. In Section \ref{optimization}, we formulate an optimization problem for OA-FL in MIMO Cloud-RAN and develop an effective algorithm to solve it. In Section \ref{coder}, we present a practical design, which exploits the inter-AP correlation. In Section \ref{results}, we validate proposed OA-FL in MIMO Cloud-RAN through experiments.

%To address the aforementioned limitations of OA-FL, we first deployed OA-FL to MIMO Cloud-RAN. Specifically, it consists of three steps: (1) edge aggregation: each access point (AP), i.e. RRH, uses over-the-air computation to aggregate the local updates transmitted by mobile devices through MIMO-based radio access network to generate edge updates; (2) global aggregation: each cloud server (CS), i.e. BBU, aggregates the edge updates transmitted by APs through the fronthaul network to generate global updates. We note that most existing Cloud-RAN or MIMO Cloud-RAN literature regards the signal transmitted by AP to CS as sampling points from independent sources and decodes the signals transmitted by each AP point independently. However, in the FL scenario, we observed that the edge updates transmitted by AP are not independent in machine learning tasks, but exhibit strong correlation. Moreover, the server does not need to estimate all edge updates, but only needs the global updates, which are a function of the edge updates. Therefore, we designed a decoder to directly estimate the global updates and consider the local updates transmitted by AP as sampling points of correlated sources. (3) Model update and broadcast: CS uses global updates to update model parameters and broadcasts the updated model parameters to all AP points through the fronthaul network. All AP points then broadcast the updated model parameters to mobile devices through MIMO-based wireless access network.

\emph{Notations}: Throughout, regular small letters, regular capital letters, bold small letters, and bold capital letters denote scalars, random variables, vectors, and matrices, respectively. $(\cdot)^\dagger$, $(\cdot)^\mathsf{T}$ and $(\cdot)^\mathsf{H}$ denote the conjugate, the transpose, and the conjugate transpose, respectively. $\mathbb{R}$ and $\mathbb{C}$ denote the real and complex number sets, respectively. $[k]$ denotes the integer set $\{1,\dots,k\}$. $[\bm{A}]_{k_1,k_2}$ denotes the $k_1$-th row and $k_2$-th column entry of matrix $\bm{A}$. Given the sets $\mathcal{K}_1\subseteq [d_1], \mathcal{K}_2\subseteq [d_2]$, $[\bm{A}]_{\mathcal{K}_1,\mathcal{K}_2}$ denotes the $|\mathcal{K}_1|\times|\mathcal{K}_1|$ matrix obtained by removing all the $(k_1, k_2)$-th elements of $\bm{A}$ with $k_1\in[d_1]\backslash \mathcal{K}_1$ or $k_2\in[d_2]\backslash \mathcal{K}_2$. When $\mathcal{K}_1=\mathcal{K}_2=\mathcal{K}$, $[\mathrm{A}]_{\mathcal{K}_1, \mathcal{K}_2}$ is simplified as $[\mathbf{A}]_{\mathcal{K}}$. $\mathcal{N}\left(\mu, \sigma^2\right)$ denotes normal distribution with mean $\mu$ and variance $\sigma^2$. $\mathcal{CN}(\mu, \sigma^2)$ denotes the circularly-symmetric complex normal distribution with mean $\mu$ and covariance $\sigma^2$. $\mathcal{N}(\boldsymbol{\mu}, \boldsymbol{\Sigma})$ denotes the multivariate normal distribution with mean vector $\boldsymbol{\mu}$ and covariance matrix $\boldsymbol{\Sigma}$. $\|\bm{x}\|$ denotes the $l_2$ norm of the vector $\bm{x}$. $\|\bm{A}\|$ denotes the induced 2-norm of the matrix $\bm{A}$. $\bm{0}$ or $\bm{1}$ denote all-zero or all-one vectors or matrices, respectively. $\bm{I}$ denotes the identity matrix.  $\nabla_{\bm{x}}=\frac{\partial}{\partial \bm{x}}$ denotes the gradient operator with respect to the vector $\bm{x}$. $\mathbb{E}[\cdot]$ denotes the expectation operator.

%$I(A; B)$ denotes the mutual information between random variables $a$ and $b$. $I(A; B\mid C)$ denotes the conditional mutual information between random variables $A$ and $B$ given $C$.

\section{Preliminaries}\label{Preliminaries}
\subsection{Federated Learning}
We consider a wireless federated learning (FL) system, where $N_D$ edge devices and a cloud server (CS) collaboratively learn a shared model based on the training data distributed over the $N_D$ edge devices. The objective is to minimize a global loss function $\mathcal{L}(\boldsymbol{\theta})$, i.e.,
\begin{equation}\label{system model0}
	\min _{\boldsymbol{\theta}} \mathcal{L}(\boldsymbol{\theta}) \triangleq \sum_{k \in\left[N_D\right]} \mathcal{L}_k(\boldsymbol{\theta}),
\end{equation}
where vector $\boldsymbol{\theta}\in\mathbb{R}^N$ is the global model parameter vector with length $N$, and $\mathcal{L}_k(\boldsymbol{\theta})$ is the empirical loss of device $k$, defined by
\begin{equation}
	\mathcal{L}_k(\boldsymbol{\theta})=\sum_{\boldsymbol{b} \in \mathcal{B}_k} l(\boldsymbol{\theta}, \boldsymbol{b}),
\end{equation}
where $l(\boldsymbol{\theta}, \boldsymbol{b})$ is the sample-wise loss function, $\mathcal{B}_k$ is the local dataset of device $k$, and $\boldsymbol{b}$ is a data sample of $\mathcal{B}_k$. 

The minimization in (\ref{system model0}) is typically solved by gradient descent (GD). Specifically, at each training round $t$, the model parameter $\bm{\theta}$ is updated via
\begin{equation}\label{model updating0}
	\bm{\theta}^{(t+1)} = \bm{\theta}^{(t)}-\eta \nabla_{\bm{\theta}}\mathcal{L}(\boldsymbol{\theta}^{(t)})= \bm{\theta}^{(t)}-\eta \sum_{k\in[N_D]}\boldsymbol{g}_k^{(t)},
\end{equation}
where $\eta\in \mathbb{R}$ is a predetermined learning rate, and $\bm{g}_k^{(t)}\triangleq \nabla_{\bm{\theta}}\mathcal{L}_k(\boldsymbol{\theta}^{(t)})\in \mathbb{R}^N$ is the local update uploaded from each device $k$ to the CS over an uplink channel. After that, $\boldsymbol{\theta}^{(t+1)}$ is broadcast to all the devices by the CS over a downlink channel to synchronize the learning model among the devices. The above process continues until it converges.

%\subsection{Over-the-Air Computation}
%Since the update in (\ref{model updating0}) for FL does not need the individual local updates generated by the devices, but a sum of them, over-the-air computation can be employed to reduce the number of channel uses involved in the uploading of local updates from the devices to the CS. Unlike orthogonal resource allocation among devices to avoid interference, over-the-air computation enables devices to share radio resources over an multiple access (MAC) channel during model uploading by leveraging the signal-superposition property of analog transmission for model aggregation over the air. This method allows for the bandwidth requirement or the communication latency to remain constant, regardless of the number of devices, which significantly mitigates the communication bottleneck in FL.

\subsection{MIMO Cloud-RAN}\label{overall channel}
In this paper, we consider the deployment of the above FL over a cloud-computing-based architecture for cellular networks, termed MIMO Cloud-RAN \cite{checko_cloud_2015}. This architecture consists of multiple access points (APs), also known as remote radio heads (RRHs). Each AP serves an individual set of mobile devices. These APs transmit (or receive) signals to (or from) mobile devices through a MIMO-based radio access network, and uploads (or downloads) the data payload to (or from) the CS via a fronthaul network, as depicted in Fig. \ref{fig:system}. Details are given below. 

\begin{enumerate}
	\item Uplink radio access network: The uplink radio access network is a MIMO MAC channel, where each AP is equipped with $N_R$ antennas, and each terminal device with  $N_T$ antennas. We assume block-fading, i.e., the channel state information (CSI) keeps invariant during a training round. Let $C$ be the number of uplink channel uses in a training round. Then the received signal matrix at AP $i$ in the $t$-th round is expressed as
	\begin{equation}\label{system:wireless}
		\boldsymbol{Y}_i^{(t)}=\sum_{k \in \mathcal{N}_{D,i}} \boldsymbol{H}_{i k}^{(t)} \boldsymbol{X}_k^{(t)}+\boldsymbol{Z}_i^{(t)}\in\mathbb{R}^{N_R\times C}, \forall i \in\left[N_A\right],
	\end{equation}
	where $\boldsymbol{X}_k^{(t)}\in{\mathbb{C}^{N_T\times C}}$ is the signal matrix transmitted by device $k$, $\boldsymbol{H}_{i k}^{(t)}\in{\mathbb{C}^{N_R\times N_T}}$ is the channel matrix between device $k$ and AP $i$\footnote{The CSI $\{\bm{H}_{ik}\}_{k\in\mathcal{N}_{D,i}},\forall i \in[N_A]$ is estimated at each AP $i$ via channel training \cite{abari_over--air_2016, dong_blind_2020}.}, $\boldsymbol{Z}_i^{(t)}\in{\mathbb{C}^{N_R\times C}}$ is an additive white Gaussian noise (AWGN) matrix whose entries are independently drawn from $\mathcal{CN}(0, {{\sigma_z}^{(t)}}^2)$, and $\mathcal{N}_{D,i}$ is the set of devices served by AP $i$ satisfying $\mathcal{N}_{D,1}\cup\dots\cup\mathcal{N}_{D,N_A}=[N_D]$ and $\mathcal{N}_{D,i}\cap\mathcal{N}_{D,j}=\varnothing,\forall i\ne j$. Note that the sum in (\ref{system:wireless}) implies that signals $\{\boldsymbol{X}_k^{(t)}\}_{k\in\mathcal{N}_{D,i}}$ are transmitted over shared radio resources and are synchronized among devices\footnote{The synchronization can be realized by using the existing techniques, e.g., the timing-advance mechanism for uplink synchronization in 4G Long Term Evolution (LTE) \cite{noauthor_3gpp_2010}.}. This enables the utilization of the superposition characteristic of electromagnetic waves for signal aggregation, a.k.a., over-the-air computation.
	
	\item Uplink fronthaul network: The uplink fronthaul network relies on dedicated fibers for data delivery \cite{park_fronthaul_2014}. Specifically, each AP $i$ transmits its signals to the CS via a rate-constrained wired link.

	\item Downlink fronthaul network: Similarly, in the downlink fronthaul network, the CS broadcasts its signals to every AP via a rate-constrained wired link.

	\item Downlink radio access network: The downlink radio access network associated with each AP $i$ is a MIMO broadcast channel, i.e., the downlink dual of the MIMO MAC in (\ref{system:wireless}). In each training round $t$, each AP $i$ broadcasts its signals to every device in $\mathcal{N}_{D,i}$.

\end{enumerate}

\begin{figure}
	\centering
	\includegraphics[width=0.45\linewidth]{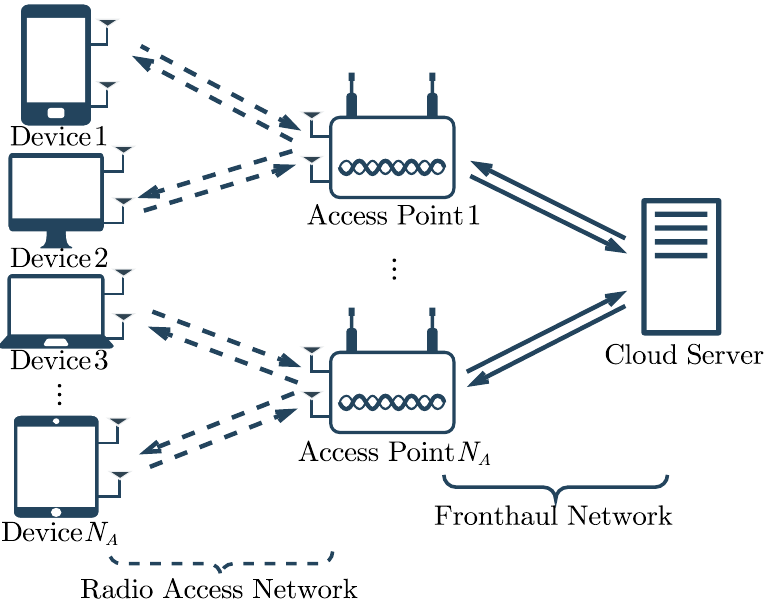}
	\caption{An illustration of the MIMO Cloud-RAN.}
	\label{fig:system}
\end{figure}

%In the edge aggregation process, each access point (AP) aggregates the local updates from devices through a MIMO MAC block-fading wireless channel, and the channel state information (CSI) remains unchanged within $C$ channel uses. Specifically, the received signal matrix at $i$-th Ap is denoted by

%are to map the local updates $\boldsymbol{g}_{k}$ to the signal matrix $\boldsymbol{X}_k$ at device $k$, to construct the edge updates $\bm{s}_i$ from the received signal $\boldsymbol{Y}_i$ at AP $i$, and to construct the global updates from $\{\bm{u}_i\}_{i=1}^{N_A}$ at the CS via a decoding function $f$.

 %The details of the transceiver design in the devices, APs and the CS are presented in what follows.

\section{OA-FL in MIMO Cloud-RAN}\label{c-ran}
We now describe the FL process deployed in the MIMO Cloud-RAN system. As depicted in Fig. \ref{fig:mysystem}, the OA-FL in MIMO Cloud-RAN framework consists of three stages: (i) edge aggregation; (ii) global aggregation; and (iii) model updating and broadcasting.

\subsection{Edge Aggregation}
We start with the edge aggregation stage of the proposed framework, where each AP aggregates local updates from $N_D$ devices through the uplink radio access network in an over-the-air fashion. Specifically, at the $t$-th round, each device $k$ first normalizes its local update $\boldsymbol{g}_{k}^{(t)}$ by
\begin{equation}\label{device:normalization}
	\tilde{\boldsymbol{g}}_k^{(t)}=\left(\boldsymbol{g}_k^{(t)}-\bar{g}_k^{(t)} \bm{1}\right)\Big/\sqrt{v_k^{(t)}}\in\mathbb{R}^N, \forall k \in\left[N_D\right],
\end{equation}
where $v_k^{(t)}=\frac{1}{N}\left\|\boldsymbol{g}_k^{(t)}-\bar{g}_k^{(t)} \mathbf{1}\right\|^2$ and $\bar{g}_k^{(t)}=\frac{1}{N} \sum_{n \in[N]}\left[\boldsymbol{g}_k^{(t)}\right]_n$\footnote{We assume that the CS perfectly knows $\bar{g}_k^{(t)}$ and $v_k^{(t)}$. In practice, the cost of transmitting the scalars $\{\bar{g}_k^{(t)}, v_k^{(t)}\}_{k\in[N_D]}$ to the CS is negligible as compared to the overall cost.}. To match the complex communication system in (\ref{system:wireless}), each device $k$ maps $\tilde{\boldsymbol{g}}_k$ to a complex vector $\bm{r}_k\in \mathbb{C}^{C}$, defined as
\begin{align}\label{device:complex}
	\operatorname{Re}\{\boldsymbol{r}_k^{(t)}\}\triangleq [\tilde{g}_{k,1}^{(t)},\dots,\tilde{g}_{k,N/2}^{(t)}]^\mathsf{T},\quad \operatorname{Im}\{\boldsymbol{r}_k^{(t)}\}&\triangleq [\tilde{g}_{k,N/2+1}^{(t)},\dots,\tilde{g}_{k,N}^{(t)}]^\mathsf{T},
\end{align}
where $\tilde{g}_{k,m}^{(t)}$ is the $m$-th entry of $\bm{\tilde{g}}_k^{(t)}$, and $C=\frac{N}{2}$. Then each device $k$ transmits $\boldsymbol{r}_k^{(t)}$ with $C$ channel uses, i.e., the transmitted signal matrix of device $k$ is given by
\begin{equation}\label{device:channel input}
	\boldsymbol{X}_k^{(t)}\triangleq\boldsymbol{\alpha}_k^{(t)} {\boldsymbol{r}_k^{(t)}}^\mathsf{T} \in \mathbb{C}^{N_T \times C}, \forall k \in\left[N_D\right],
\end{equation}
where $\boldsymbol{\alpha}_k^{(t)}\in\mathbb{C}^{N_T}$ is transmit beamforming vector of device $k$ with $N_T$ being the number of transmit antennas. Let $r_{k,c}^{(t)}$ be the $c$-th entry of $\boldsymbol{r}_k^{(t)}$ and $\bm{x}_{k,c}^{(t)}\in\mathbb{C}^{N_T}$ be the $c$-th col of $\boldsymbol{X}_k^{(t)}$. The transmitted signal of the device $k$ at the $c$-th channel use is given by $\bm{x}_{k,c}^{(t)}=r_{k,c}^{(t)}\bm{\alpha}_k^{(t)}$, satisfying the power constraint of
\begin{equation}\label{device:power constraint}
	\mathbb{E}\left[\left\|\boldsymbol{x}_{k,c}^{(t)}\right\|^2\right]=2\left\|\boldsymbol{\alpha}_k^{(t)}\right\|^2 \leq P_k^{(t)}, \forall k \in\left[N_D\right], \forall n \in[N],
\end{equation}
where the equality is from (\ref{device:normalization}) (implying $\mathbb{E}[|{r}_{k,c}^{(t)}|^2]=2$), and $P_k$ is the power budget of device $k$. After transmitting the signal matrices $\{\boldsymbol{X}_k^{(t)}\}_{k\in\mathcal{N}_{D,i}}$ to AP $i$ via the uplink radio access network in (\ref{system:wireless}), AP $i$ applies a receive beamforming vector to combine the received signal matrix $\bm{Y}_i^{(t)}$ over the air, resulting in $\hat{\boldsymbol{r}}_i^{(t)}$ as
\begin{equation}\label{acc:receive}
	\hat{\boldsymbol{r}}_i^{(t)}=\left({\boldsymbol{\beta}_i^{(t)}}^{\mathrm{H}} \boldsymbol{Y}_i^{(t)}\right)^\mathsf{T}=\sum_{k \in \mathcal{N}_{D,i}} \boldsymbol{r}_k^{(t)}\left({\boldsymbol{\beta}_i^{(t)}}^{\mathrm{H}} \boldsymbol{H}^{(t)}_{i k} \boldsymbol{\alpha}_k^{(t)}\right)^\mathsf{T}+{\boldsymbol{Z}^{(t)}_i}^\mathsf{T} {\boldsymbol{\beta}_i^{(t)}}^{\dagger} \in \mathbb{C}^{C}, \forall i \in\left[N_A\right],
\end{equation}  
where $\boldsymbol{\beta}_i^{(t)}\in\mathbb{C}^{N_R}$ is the receive beamforming vector of AP $i$. Then each AP $i$ constructs its edge update as
\begin{equation}\label{acc:process}
	{\boldsymbol{s}}_i^{(t)}\triangleq\left[\begin{array}{l}
		\operatorname{Re}\left\{\hat{\boldsymbol{r}}_i^{(t)}\right\}^\mathsf{T} ,
		\operatorname{Im}\left\{\hat{\boldsymbol{r}}_i^{(t)}\right\}^\mathsf{T}
	\end{array}\right]^\mathsf{T} \in \mathbb{R}^N, \forall i \in\left[N_A\right].
\end{equation}

\begin{figure}
	\centering
	\includegraphics[width=0.95\linewidth]{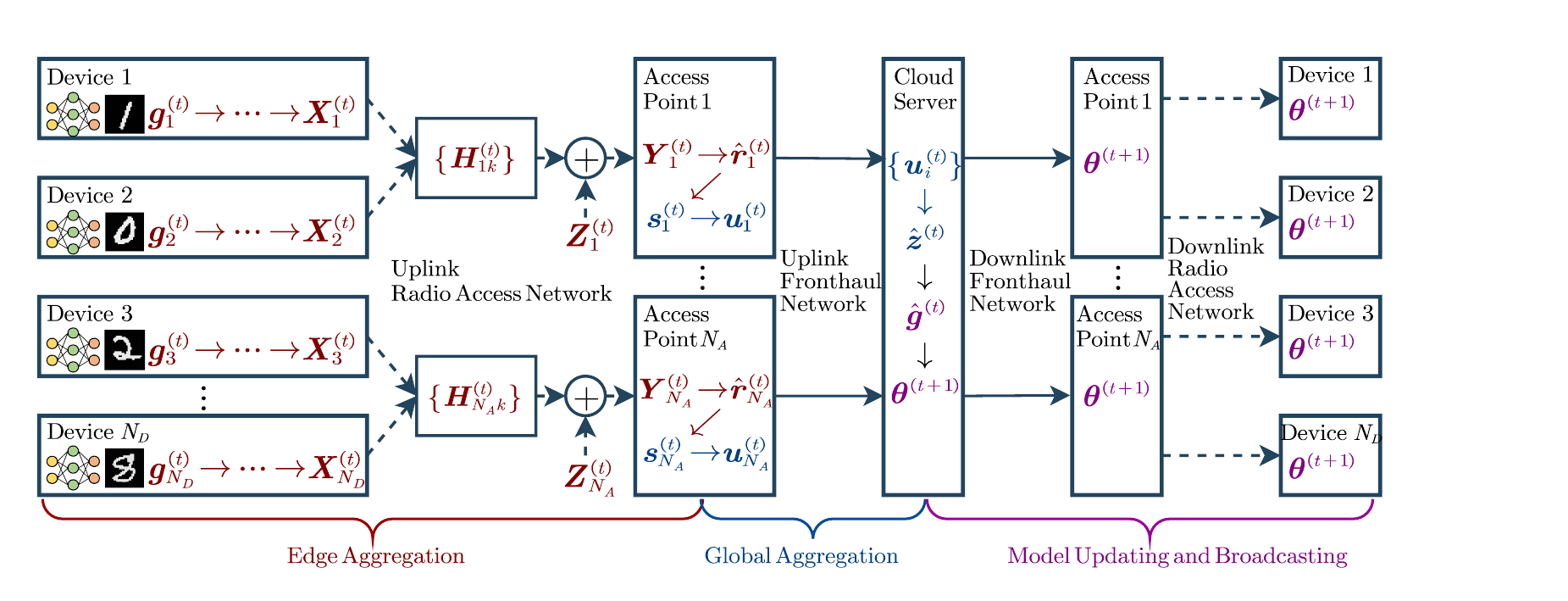}
	\caption{An illustration of a training round of MIMO Cloud-RAN based OA-FL.}
	\label{fig:mysystem}
\end{figure}

%Then each device $k$ transmits the signal matrix $\bm{X}_k$ to the APs. Since the update in (\ref{model updating0}) does not require individual local updates from the devices, but rather a sum of them, we utilize MIMO-based over-the-air computation for the edge aggregation. This method has the advantage of not increasing the bandwidth requirement or communication latency with the number of devices, which largely relieves the communication bottleneck\footnote{In general, over-the-air aggregation assumes symbol-level synchronization among the transmitted devices, so as to align the arrived signals in the signal superposition for model aggregation. The synchronization can be realized by using the existing techniques, e.g., the timing advance (TA) mechanism for uplink synchronization in 4G Long Term Evolution (LTE) [xxx]}. 

%We assume that the AP $i$ knows the CSI $\{\bm{H}_{ik}\}_{k=1}^{N_D}$ at each round. In practice, the CSI can be acquired by the APs through channel training [xxx]\footnote{The learning performance of the over-the-air aggregation based on analog transmission has been shown to be robust against imperfect CSI at the devices [xxx]}.

\subsection{Global Aggregation}\label{global agg}
We now describe the global aggregation stage of our proposed framework, where the CS aggregates edge updates from $N_A$ APs through the uplink fronthaul network. Specifically, at the $t$-th round, each AP $i$ encodes $\bm{s}_i^{(t)}$ to a vector
\begin{equation}\label{encoder}
	\bm{u}_i^{(t)}=f_i^{(t)}\left(\bm{s}_i^{(t)}\right)\in[B_i^{(t)}],\forall i\in[N_A],
\end{equation}
where $f_i^{(t)}:\mathbb{R}^N \rightarrow [B_i^{(t)}], \forall i\in[N_A]$ is an encoding function, and $B_i^{(t)}$ is the codebook size of $f_i^{(t)}$. Then each AP transmits $\bm{u}_i^{(t)}$ to the CS over the uplink fronthaul network. Upon receiving the signals $\{\bm{u}_i^{(t)}\}_{i\in[N_A]}$, the CS aggregates (or alternatively, decodes) them as a global update.
\begin{equation}\label{decoder}
	\hat{\bm{z}}^{(t)}=f^{(t)}\left(\bm{u}_1^{(t)},\dots,\bm{u}_{N_A}^{(t)}\right) \in\mathbb{R}^{N},
\end{equation}
where $f^{(t)}:[B_i^{(t)}]\times \dots \times [B_{N_A}^{(t)}] \rightarrow \mathbb{R}^N$ is a decoding function. Local updates in FL are typically correlated \cite{zhong_over--air_2021}, resulting in correlated edge updates. This correlation, called inter-AP correlation, can be exploited to significantly reduce the communication cost of global aggregation. To improve communication efficiency, we treat global aggregation as a lossy distributed source coding (L-DSC) process. In L-DSC, given an aggregation target function
$
	\kappa:\underbrace{\mathbb{R}^N\times \dots \times \mathbb{R}^N}_{N_A} \rightarrow\mathbb{R}^N
$
 and a distortion measure function
$
	d : \mathbb{R}^N \times \mathbb{R}^N \rightarrow \mathbb{R}_{+},
$ an achievable rate-distortion tuple is defined below.
\begin{defi}\label{defi1}
	A rate-distortion tuple $\left(R_1^{(t)}, \ldots, R_{N_A}^{(t)}, D^{(t)}\right), \forall t$ is said to be achievable if for any $\epsilon>0$ and any sufficiently large $N$, there exists $N_A$ encoding functions $\left\{f_i^{(t)}\right\}_{i=1}^{N_A}$ and a joint decoding function $f^{(t)}$ such that rate $\frac{1}{N} \log \left(B_i^{(t)}\right) \leq R_i^{(t)}+\epsilon, \forall i \in[N_A]$, and expected aggregation distortion $\frac{1}{N}\mathbb{E}\left[d\left(\mathbf{z}^{(t)}, \hat{\mathbf{z}}^{(t)}\right)\right] \leq D^{(t)}+\epsilon$, where $\bm{z}^{(t)}\triangleq\kappa\left(\bm{s}_1^{(t)},\dots,\bm{s}_{N_A}^{(t)}\right)$.
\end{defi}
Loosely speaking, as the model dimension $N$ increases, the expected distortion $\frac{1}{N}\mathbb{E}\left[d\left(\mathbf{z}^{(t)}, \hat{\mathbf{z}}^{(t)}\right)\right]$ can be arbitrarily close to $D^{(t)}$, when the rate $\frac{1}{N}\log(B_i^{(t)})$ tends to $R_i^{(t)},\forall i\in[N_A]$. The rate-distortion region of L-DSC is defined as the set of all achievable rate-distortion tuples, denoted by $\mathcal{RD}^{(t)}$. Detailed analysis of L-DSC will be provided later in Section \ref{analysis L-DSC}.

%The global aggregation involving the encoding, transmission and decoding of the edge updates $\bm{s}_i^{(t)}$ can be formulated as a lossy distributed source coding (L-DSC) problem \cite{yang_federated_2022}.

%Note that the global aggregation involves encoding $N_A$ correlated information sources separately by the APs, transmitting them over rate-constrained, noiseless channels to a signal destination, and then receiving by the CS for reconstructing a global update, 

%we thus consider the global aggregation as a lossy distributed source coding (L-DSC) problem [xxx].

%The primary objective of the L-DSC problem is to design the functions $\{f_i\}_{i=1}^{N_A}, f$ that minimizes the distortion $d(\hat{\bm{z}}, \bm{z})$, while satisfying a rate-constraint 

%\begin{equation}\label{rate_con}
%	\frac{1}{N}\log(B_i)\leq R_i,
%\end{equation}
%where $R_i$ denotes the rate in units of bit/symbol, and $d:\mathbb{R}^N\times\mathbb{R}^N\rightarrow\mathbb{R}_{+}$ represents a distortion measure between two N-dimensional vectors.

%However, there are various specific practice of $f_i$, making it difficult to establish a unified analysis framework. 
%可以证明，当我们施加一些常规的合理的假设在$\bm{s}_i$上时，$\{\bm{A}\bm{s}_i\}_{i=1}^{N_A}$可以被建模为一些无记忆相关高斯信源的采样，并且$\{\bm{u}_i\}_{i=1}^{N_A}$可以被建模为两个高斯向量的叠加。

\begin{algorithm}
	\caption{OA-FL in MIMO Cloud-RAN.
	}\label{algo:system}
	\begin{algorithmic}[1] %这个1 表示每一行都显示数字 	
		\STATE \textbf{At each training round $t$}
		\STATE \textbf{Edge aggregation:}
		\STATE \quad Each device $k$ computes its local update $\bm{g}_k^{(t)}$ with $\mathcal{B}_k$ and $\bm{\theta}$
		\STATE \quad Each device $k$ computes $\bm{X}_k^{(t)}$ from $\bm{g}_k^{(t)}$ via (\ref{device:normalization})-(\ref{device:channel input})
		\STATE \quad Each device $k$ transmits $\bm{X}_k^{(t)}$ to the APs via the channel in (\ref{system:wireless}) 
		\STATE \quad Each AP $i$ receives the matrix $\bm{Y}_i^{(t)}$ and computes $\bm{s}_i^{(t)}$ via (\ref{acc:receive})-(\ref{acc:process})
		
		\STATE \textbf{Global aggregation:}
		\STATE \quad Each AP $i$ computes $\bm{u}_i^{(t)}$ from $\bm{s}_i^{(t)}$ via (\ref{encoder})
		\STATE \quad Each AP $i$ transmits $\bm{u}_i^{(t)}$ to the CS
		\STATE \quad The CS receives $\{\bm{u}_i^{(t)}\}_{i=1}^{N_A}$ and computes $\hat{\bm{z}}^{(t)}$ via (\ref{decoder})

		\STATE \textbf{Model updating and broadcast:}
		\STATE \quad The CS computes $\hat{\boldsymbol{g}}^{(t)}$ from $\hat{\bm{z}}^{(t)}$ via (\ref{cs:gradient estimation}) and updates $\bm{\theta}^{(t)}$ via (\ref{cs:model update})
		\STATE \quad The CS broadcasts $\bm{\theta}^{(t+1)}$ to all the APs
		\STATE \quad Each AP $i$ broadcasts the received $\bm{\theta}^{(t+1)}$ to every device in $\mathcal{N}_{D,i}$
	\end{algorithmic} 	
\end{algorithm}

\subsection{Model Updating and Broadcasting}
We now describe the model updating and broadcasting stage of our proposed framework. Specifically, the CS computes the vector $\hat{\bm{g}}^{(t)}\in \mathbb{R}^N$ as
\begin{equation}\label{cs:gradient estimation}
	\hat{\boldsymbol{g}}^{(t)}= \hat{\boldsymbol{z}}^{(t)}+\sum_{k \in\left[N_D\right]} \bar{g}_k^{(t)} \mathbf{1}\in\mathbb{R}^N,
\end{equation}
where the bias term $\sum_{k\in[N_D]}\bar{g}_k^{(t)}$ given in (\ref{device:normalization}) is known at the CS through error-free transmission. Then the CS updates the global model as
\begin{equation}\label{cs:model update}
	\boldsymbol{\theta}^{(t+1)} = \boldsymbol{\theta}^{(t)}-\eta \hat{\boldsymbol{g}}^{(t)},
\end{equation}
where $\eta$ is given in (\ref{model updating0}). Then the CS broadcasts the global model $\boldsymbol{\theta}^{(t+1)}$ to every AP via the downlink fronthaul network described in Section \ref{overall channel}. Next each AP $i$ broadcasts received $\boldsymbol{\theta}^{(t+1)}$ to every device in $\mathcal{N}_{D,i}$. Following the common practice in \cite{liu_reconfigurable_2021,amiri_federated_2020,zhong_over--air_2021,yu_optimizing_2020,yang_federated_2022}, we assume that the broadcast process in the downlink fronthaul network and the downlink radio access network is error-free. The overall scheme of OA-FL in MIMO Cloud-RAN is summarized in Algorithm \ref{algo:system}.

\section{Performance Analysis}\label{analysis}
In this section, we first analyze the L-DSC problem described in Section \ref{global agg}, and then carry out the performance analysis of the MIMO Cloud-RAN based OA-FL, yielding an upper bound on the learning performance.

\subsection{Preliminaries}\label{analysis L-DSC}
%from a rate-distortion theory perspective .

%In the following, we focus on an arbitrary communication round $t$ and omit the superscript $(t)$ for brevity whenever causing no ambiguity.

Recall that we model global aggregation as an L-DSC process. We define the rate-distortion region $\mathcal{RD}^{(t)}$ as the set of all achievable rate-distortion tuples in L-DSC at the $t$-th training round. However, this region is generally unknown. To establish a tractable inner region of $\mathcal{RD}^{(t)}$, we rely on the following assumptions.

\begin{assumption}\label{gaussian ass}
	Denote a matrix $\bm{\Sigma}_S^{(t)}\in\mathbb{R}^{N_A\times N_A}$ with $[\bm{\Sigma}_S^{(t)}]_{i,j}=\frac{1}{N}\mathbb{E}\left[(\bm{s}_i^{(t)})^\mathsf{T}\bm{s}_j^{(t)}\right],\forall i,j\in[N_A]$. Denote the $n$-th entry of $\bm{s}_i^{(t)}$ as $s_{i,n}^{(t)}$. For the edge updates $\{\bm{s}_i^{(t)}\}_{i\in[N_A]}$, we assume that (i) $s_{i,n}^{(t)}\sim\mathcal{N}(0,[\bm{\Sigma}_S^{(t)}]_{i,i}), \forall i\in[N_A],  \forall n\in[N]$, and (ii) $(s_{1,n}^{(t)}$, $s_{2,n}^{(t)}$, $\dots$, $s_{N_A,n}^{(t)})\sim\mathcal{N}(\bm{0},\bm{\Sigma}_S^{(t)}), \forall   n\in[N]$.
\end{assumption}
Assumption \ref{gaussian ass} allows us to treat $\{s_{1,n}^{(t)}$, $\dots$, $s_{N_A,n}^{(t)}\}_{n\in[N]}$ as $N$ samples generated from an $N_A$-component memoryless Gaussian source $\mathcal{N}(\bm{0},\bm{\Sigma}_S^{(t)})$\footnote{The matrix $\bm{\Sigma}_S$ models correlations among APs, and is either known a priori or estimated during training. A simple approach is for each AP to send a sub-vector of $\bm{s}_i$ to the CS intermittently for parameter estimation during training.}. Based on Assumption \ref{gaussian ass}, we provide an inner region of $\mathcal{RD}^{(t)}$ as follows:
\begin{lemma}\label{RDin}
	Under the quadratic distortion measure $d(\mathbf{a}^{(t)}, \mathbf{b}^{(t)})=\|\mathbf{a}^{(t)}-\mathbf{b}^{(t)}\|^2$ and the linear aggregation target $\kappa\left(\bm{a}_1^{(t)}, \ldots, \bm{a}_{N_A}^{(t)}\right)=\sum_{i=1}^{N_A} c_i^{(t)} \bm{a}_i^{(t)}$, where $c_i^{(t)}$ is the $i$-th entry of a weighted coefficient vector $\bm{c}^{(t)}\in\mathbb{R}^{N_A}$, it is concluded from Assumption \ref{gaussian ass} that
\begin{align}
	\nonumber\mathcal{R} \mathcal{D}_{\text{in}}^{(t)}&=\bigcup_{\mathbf{\Sigma}_V^{(t)}}\bigg\{\left(R_1^{(t)}, \ldots, R_{N_A}^{(t)}, D^{(t)}\right): \\
	\nonumber& (i) \sum_{i \in \mathcal{K}} R_i^{(t)} \geq \frac{1}{2} \log \left(\frac{\operatorname{det}\left(\bm{\Sigma}_S^{(t)}+\bm{\Sigma}_V^{(t)}\right)}{\operatorname{det}\left([\bm{\Sigma}_S^{(t)}]_{\mathcal{K}^c}+[\bm{\Sigma}_V^{(t)}]_{\mathcal{K}^c}\right) \operatorname{det}\left([\bm{\Sigma}_V^{(t)}]_{\mathcal{K}}\right)}\right), \forall \mathcal{K} \subset[N_A], \mathcal{K} \neq \varnothing; \\
	\nonumber& (ii) \sum_{i \in[N_A]} R_i^{(t)} \geq \frac{1}{2} \log \left(\frac{\operatorname{det}\left(\bm{\Sigma}_S^{(t)}+\bm{\Sigma}_V^{(t)}\right)}{\operatorname{det}\left(\bm{\Sigma}_V^{(t)}\right)}\right);\\
	& (iii) D^{(t)} \geq \frac{1}{N}\mathbb{E}[\|\bm{z}^{(t)}-\hat{\bm{z}}^{(t)}\|^2]\bigg\},
\end{align}
where $\mathcal{R}\mathcal{D}_{\text{in}}^{(t)}\subset\mathcal{R}\mathcal{D}^{(t)}$, rate $R_i^{(t)},\forall i$ and distortion $D^{(t)}$ are given in Definition \ref{defi1}, $\hat{\bm{z}}$ is given in (\ref{decoder}), $\bm{z}^{(t)}\triangleq\kappa\left(\bm{s}_1^{(t)},\dots,\bm{s}_{N_A}^{(t)}\right)$, $\boldsymbol{\Sigma}_V^{(t)}\in\mathbb{R}^{N_A\times N_A}$ is an undetermined diagonal matrix called L-DSC parameter, and $\frac{1}{N}\mathbb{E}[\|\bm{z}^{(t)}-\hat{\bm{z}}^{(t)}\|^2]\triangleq{\mathbf{c}^{(t)}}^{\mathsf{T}} \boldsymbol{\Sigma}_S^{(t)} \mathbf{c}^{(t)}-{\mathbf{c}^{(t)}}^{\mathsf{T}} \boldsymbol{\Sigma}_S^{(t)}\left(\boldsymbol{\Sigma}_S^{(t)}+\boldsymbol{\Sigma}_V^{(t)}\right)^{-1} {\boldsymbol{\Sigma}_S^{(t)}}^{\mathsf{T}} \mathbf{c}^{(t)}$. 
\end{lemma}
\begin{proof}
	See the proof in \cite{yang_federated_2022}.
\end{proof}
We emphasize that the choice of L-DSC parameter $\boldsymbol{\Sigma}_V^{(t)}$ leads to different codebooks for L-DSC. By tuning the L-DSC parameter $\boldsymbol{\Sigma}_V^{(t)}$, we search for the rate-distortion tuple in $\mathcal{R} \mathcal{D}_{\text{in}}^{(t)}$ that minimizes the distortion $\mathbb{E}[\|\bm{z}^{(t)}-\bm{\hat{z}}^{(t)}\|^2]$. However, due to the link rate-constraints, $\boldsymbol{\Sigma}_V^{(t)}$ cannot be chosen arbitrarily. Let $r_i^{(t)}$ denote the maximum information rate that AP $i$ can transmit to the CS at the training round $t$. To ensure reliable uplink transmission, the (source coding) rates in Lemma \ref{RDin} must satisfy
\begin{equation}\label{bit con}
	R_i^{(t)}\leq r_i^{(t)},\forall i\in[N_A].
\end{equation}

\subsection{Communication Error Analysis of OA-FL in MIMO Cloud-RAN}
We now analyze the communication error in the OA-FL in MIMO Cloud-RAN framework. The gradient aggregation error at the $t$-th training round is given by
\begin{equation}\label{ana:error}
	\bm{e}^{(t)}=\nabla \mathcal{L}(\bm{\theta}^{(t)})-
	\hat{\boldsymbol{g}}^{(t)}\in\mathbb{R}^{N},
\end{equation}
where $\nabla \mathcal{L}(\bm{\theta}^{(t)})$ is from (\ref{model updating0}) and $\hat{\boldsymbol{g}}^{(t)}$ is from (\ref{cs:model update}). From (\ref{device:normalization})-(\ref{cs:gradient estimation}), we see that the error $\bm{e}$ mainly arises from the transmission in the uplink radio access network, and the transmission in the uplink fronthaul network. Correspondingly, we decompose the error $\bm{e}$ into two terms as
\begin{subequations}\label{error analysis}
	\begin{align}
		\boldsymbol{e}^{(t)} &=\sum_{k \in\left[N_D\right]}\bm{g}_k^{(t)}-\sum_{k \in\left[N_D\right]} \bar{g}_k^{(t)} \mathbf{1}-\hat{\boldsymbol{z}}^{(t)}\\
		&=\underbrace{\sum_{k \in\left[N_D\right]} \sqrt{v_k^{(t)}} \tilde{\boldsymbol{g}}_k^{(t)}-\sum_{i \in\left[N_A\right]}[\boldsymbol{c}^{(t)}]_i \bm{s}_i^{(t)}}_{\text{Wireless transmission error }}+\underbrace{\bm{z}^{(t)}-\hat{\boldsymbol{z}}^{(t)}}_{\text{Wired transmission error}} \\
		&=\boldsymbol{e}_1^{(t)}+\boldsymbol{e}_2^{(t)}\in\mathbb{R}^N,
	\end{align}
\end{subequations}
where the first equality is from (\ref{cs:gradient estimation}), the second equality is from (\ref{device:normalization}), $\bm{e}_{1}^{(t)}\in\mathbb{R}^N$ is the wireless transmission error caused by the communication noise in (\ref{system:wireless}), and $\bm{e}_{2}^{(t)}\in\mathbb{R}^N$ is the wired transmission error caused by the rate-constraint in the fronthaul network. We analyze the bound of $\mathbb{E}[||\boldsymbol{e}^{(t)}||^2]$ in the following lemma.
\begin{lemma}\label{all_bound}
	Assume $\frac{1}{N}\mathbb{E}[\|(\tilde{\boldsymbol{g}}_{k_1}^{(t)})^{\mathsf{H}}{\tilde{\boldsymbol{g}}_{k_2}^{(t)}}\|^2] = [\bm{\Sigma_S}^{(t)}]_{i_1,i_2}, \forall k_1\in\mathcal{N}_{D,i_1}, k_2\in\mathcal{N}_{D,i_2}$. Then
	\begin{align}\label{error_term}
		\nonumber\mathbb{E}[||\boldsymbol{e}^{(t)}||^2]\leq 2ND_{\text{system}}^{(t)},
	\end{align}
	where
	\begin{align}
		\nonumber D_{\text{system}}^{(t)}\triangleq &\  \sum_{i_1\in[N_A]}\sum_{i_2\in[N_A]}[\bm{\Sigma_S}^{(t)}]_{i_1,i_2}\sum_{k\in\mathcal{N}_{D,i_1}}{\left(\sqrt{v_{k}^{(t)}}-c_{i_1}^{(t)}{\bm{\alpha}_{k}^{(t)}}^\mathsf{T}{\bm{H}_{i_1k}^{(t)}}^\mathsf{T}{\bm{\beta}_{i_1}^{(t)}}^\mathsf{\dag}\right)}^\mathsf{H}\\
		\nonumber&\times\sum_{k\in\mathcal{N}_{D,i_2}}\left(\sqrt{v_{k}^{(t)}}-c_{i_2}^{(t)}{\bm{\alpha}_{k}^{(t)}}^\mathsf{T}{\bm{H}_{i_2k}^{(t)}}^\mathsf{T}{\bm{\beta}_{i_2}^{(t)}}^\mathsf{\dag}\right) + \sum_{i\in[N_A]}\varepsilon_i^{(t)} {c^{(t)}_i}^2\|\bm{\beta}_i^{(t)}\|^2\\
		& +  {\boldsymbol{c}^{(t)}}^\mathsf{T} \boldsymbol{\Sigma}_S^{(t)} \boldsymbol{c}^{(t)}-{\boldsymbol{c}^{(t)}}^\mathsf{T} \boldsymbol{\Sigma}_S^{(t)}\left(\boldsymbol{\Sigma}_S^{(t)}+\boldsymbol{\Sigma}_V^{(t)}\right)^{-1} {\boldsymbol{\Sigma}_S^{(t)}}^\mathsf{T} \boldsymbol{c}^{(t)}\in\mathbb{R},
	\end{align}
	the matrices $\bm{\Sigma}_S^{(t)}$ and $\bm{\Sigma}_V^{(t)}\in\mathbb{R}^{N_A\times N_A}$ are given in Assumption \ref{gaussian ass}, $\varepsilon_i\triangleq\frac{1}{N}\mathbb{E}[\|\bm{Z}_i^{(t)}\|^2]\in\mathbb{R}, \forall i \in [N_A]$ is the noise power\footnote{The noise power $\varepsilon_i$ is estimated at each AP $i$ by the expectation-maximization algorithm \cite{vila_expectation-maximization_2013}.}, $\bm{\alpha}_k^{(t)}, \bm{\beta}_i^{(t)}, \bm{H}_{ik}^{(t)}, \forall i,k, t$ are given in (\ref{system:wireless}), $\bm{c}^{(t)}$ is given in Lemma \ref{RDin}, and $\{v_k^{(t)}\}_{k=1}^{N_D}$ are given in (\ref{device:normalization}). 
\end{lemma}

\begin{proof}
See Appendix \ref{app0}.
\end{proof}

\subsection{Convergence Analysis of OA-FL in MIMO Cloud-RAN}
We now present convergence analysis to establish the relationship between the FL convergence rate and the communication error (\ref{ana:error}). Following \cite{friedlander_hybrid_2012}, we introduce standard assumptions in stochastic optimization below.
\begin{assumption}\label{sec:conver,asu:assumption 1}
	(i)	$\mathcal{L}(\cdot)$ is strongly convex with some (positive) parameter $\mu$. That is, $\mathcal{L}(\bm{y}) \geq \mathcal{L}(\bm{x})+(\bm{y}-\bm{x})^{T} \nabla \mathcal{L}(\bm{x})+\frac{\mu}{2} \| \bm{y}-\bm{x} \|^{2}, \forall \bm{x}, \bm{y} \in \mathbb{R}^{N}$. (ii) The gradient $\nabla \mathcal{L}(\cdot)$ is Lipschitz continuous with some (positive) parameter $L$. That is, $\|\nabla \mathcal{L}(\bm{x})-\nabla \mathcal{L}(\bm{y})\| \leq L\|\bm{x}-\bm{y}\|, \forall \bm{x}, \bm{y} \in \mathbb{R}^{N}$. (iii) $\mathcal{L}(\cdot)$ is twice continuously differentiable. (iv) The gradient with respect to any training sample, denoted by $\nabla l(\bm{\theta}; \cdot)$, is upper bounded at $\bm{\theta}$ as $\nonumber\left\|\nabla l(\bm{\theta},\bm{b})\right\|^{2} \leq \gamma_{1}+ \gamma_{2} \left\|\nabla \mathcal{L}\left(\bm{\theta}\right)\right\|^{2}$, where $\bm{b}$ is a data sample from the whole datasets $\{\mathcal{B}_k\}_{k=1}^{N_D}$, and $\gamma_{1} \geq 0, \gamma_{2}>0$ are some constants.
\end{assumption}  	 	
Assumption \ref{sec:conver,asu:assumption 1}-(i) ensures that a global optimum exists for the loss function $\mathcal{L}(\cdot)$. Assumption \ref{sec:conver,asu:assumption 1}-(ii) is imposed on the gradient function $\nabla\mathcal{L}(\cdot)$ to prevent the function value from changing too fast, so as to construct a more robust machine learning model. Assumption \ref{sec:conver,asu:assumption 1}-(iv) provides a bound on the norm of the local gradients, i.e., $\bm{g}_{k}$. Assumption \ref{sec:conver,asu:assumption 1} leads to an upper bound on the loss function $\mathcal{L}(\bm{\theta}^{(t+1)})$ with respect to the model updating (\ref{cs:model update}), as given in the following lemma.
\begin{lemma}\label{lemma1}
	Assume that $\mathcal{L}(\cdot)$ satisfies  Assumption \ref{sec:conver,asu:assumption 1},
	at the $t$-th training round with 
	the learning rate $\eta$ is set to $1/L$. Then
	\begin{align} \label{equ:lemma}
		\mathbb{E}[\mathcal{L}(\bm{\theta}^{(t+1)})] \leq  	\mathbb{E}[\mathcal{L}(\bm{\theta}^{(t)})]-\frac{1}{2L}	\mathbb{E}[\|\nabla \mathcal{L}(\bm{\theta}^{(t)})\|^{2}]+\frac{1}{2L}	\mathbb{E}[ \|\bm{e}^{(t)}\|^{2}],
	\end{align}
	where the Lipschitz constant $L$ is given in Assumption \ref{sec:conver,asu:assumption 1}- (ii).
\end{lemma}
\begin{proof}
	See \cite[Lemma 2.1]{friedlander_hybrid_2012}.
\end{proof}

We are now ready to present an upper bound of the expected difference between the training loss and the optimal loss at round $t+1$, i.e., $\mathbb{E}[\mathcal{L}(\bm{\theta}^{(t+1)})-\mathcal{L}(\bm{\theta}^{(\star)})]$.
\begin{thm}\label{thm:1}
	Based on Assumption  \ref{sec:conver,asu:assumption 1}, the loss function at the $t$-th round $\mathcal{L}(\bm{\theta})^{(t)}$ satisfies
	\begin{equation}\label{equ:thm1}
		\mathbb{E}[\mathcal{L}(\bm{\theta}^{(t+1)})] -\mathbb{E}[\mathcal{L}(\bm{\theta}^{(\star)})]\leq \left(1-\frac{\mu}{L}\right)^{t}\left(\mathbb{E}[\mathcal{L}(\bm{\theta}^{(1)})] -\mathbb{E}[\mathcal{L}(\bm{\theta}^{(\star)})]\right)+\sum_{t'=1}^t \left(1-\frac{\mu}{L}\right)^{t-t'}\frac{N}{L}D_\mathsf{system}^{(t')},
	\end{equation}
where $(\cdot)^{t}$ denotes the $t$-th power operator, $\mathcal{L}(\cdot)$ is the total empirical loss function given in (\ref{system model0}), $\bm{\theta}^{(1)}$ is the initial system model parameter, and $D_\mathsf{system}^{(t)}$ is given in Lemma \ref{all_bound}.
\end{thm}
\begin{proof}
	From \cite[eq. (2.4)]{friedlander_hybrid_2012}, we have 
	\begin{equation}\label{equ:lemma 2.4}
		\mathbb{E}[||\nabla \mathcal{L}(\bm{\theta}^{(t)})||^{2}] \geq 2 \mu(\mathbb{E}[\mathcal{L}(\bm{\theta}^{(t)})]-\mathbb{E}[\mathcal{L}(\bm{\theta}^{(\star)})]).
	\end{equation}
	Subtracting $\mathbb{E}[\mathcal{L}(\bm{\theta}^{(\star)})]$ on both sides of (\ref{equ:lemma}) in Lemma \ref{lemma1} and plugging (\ref{equ:lemma 2.4}) and Lemma \ref{all_bound} into (\ref{equ:lemma}), we obtain
	\begin{equation}\label{equ:theta_star_one}
		\mathbb{E}[\mathcal{L}(\bm{\theta}^{(t+1)})]-\mathbb{E}[\mathcal{L}(\bm{\theta}^{(\star)})]\leq \left(1-\frac{\mu}{L}\right)
		\left(\mathbb{E}[\mathcal{L}(\bm{\theta}^{(t)})]-\mathbb{E}[\mathcal{L}(\bm{\theta}^{(\star)})]\right)+\frac{N }{L}D_{\mathsf{system}}^{(t)}.
	\end{equation}
	Applying the above inequality recursively yields (\ref{equ:thm1}), which completes the proof.
\end{proof}

From Theorem \ref{thm:1}, we see that $\mathbb{E}[\mathcal{L}(\bm{\theta}^{(t)})-\mathcal{L}(\bm{\theta}^{(\star)})]$, i.e., the expected difference between the training loss and the optimal loss at the $t$-th round, is upper bounded by the right-hand side of the inequality in (\ref{equ:thm1}). Moreover, this upper bound converges with speed $1-\frac{\mu}{L}$, since $\frac{\mu}{L}>0$. Empirically, we find that our proposed scheme always converges with appropriately chosen system parameters. 
The following corollary further characterizes the convergence behavior of $\mathbb{E}[\mathcal{L}(\bm{\theta}^{(t+1)})-\mathcal{L}(\bm{\theta}^{(\star)})]$.

\begin{coro}\label{coro1}
	As $t\rightarrow\infty$, we have 
	\begin{equation}\label{final_convegence_gap}
		\lim_{t \rightarrow \infty}\mathbb{E}[\mathcal{L}(\bm{\theta}^{(t+1)})-\mathcal{L}(\bm{\theta}^{(\star)})]\leq \sum_{t'=1}^t \left(1-\frac{\mu}{L}\right)^{t-t'}\frac{N}{L}D_{\mathsf{system}}^{(t')}
	\end{equation}
\end{coro}
\begin{proof}
	By noting $\frac{\mu}{L}>0$, we have $\lim_{t\rightarrow\infty}(1-\frac{\mu}{L})^t=0$. Plugging this result into (\ref{equ:thm1}), we obtain (\ref{final_convegence_gap}).
\end{proof}

Corollary \ref{coro1} shows that our proposed scheme guarantees to converge. However, there generally exists a gap between the converged loss $\lim_{t \rightarrow \infty}\mathbb{E}[\mathcal{L}(\bm{\theta}^{(t+1)})]$ and the optimal one $\mathbb{E}[\mathcal{L}(\bm{\theta}^{(\star)})]$ because of the communication error. Therefore, we next aim to minimize the gap, so as to optimize our proposed framework.

\section{System Optimization}\label{optimization}
\subsection{Problem Formulation}
From Theorem \ref{thm:1} and Corollary \ref{coro1}, the upper bound $\sum_{t'=1}^t \left(1-\frac{\mu}{L}\right)^{t-t'}\frac{N}{L}D_{\mathsf{system}}^{(t')}$ in (\ref{final_convegence_gap}) represents the impact of the wireless transmission error and the wired transmission error on the asymptotic learning performance. Specifically, a smaller $\sum_{t'=1}^t \left(1-\frac{\mu}{L}\right)^{t-t'}\frac{N}{L}D_{\mathsf{system}}^{(t')}$ leads to a smaller gap in $\lim_{t \rightarrow \infty}\mathbb{E}[\mathcal{L}(\bm{\theta}^{(t+1)})-\mathcal{L}(\bm{\theta}^{(\star)})]$. This motivates us to use $\sum_{t'=1}^t \left(1-\frac{\mu}{L}\right)^{t-t'}\frac{N}{L}D_{\mathsf{system}}^{(t')}$ as the performance metric of our proposed framework, i.e., to minimize $D_{\mathsf{system}}^{(t)}$ in Theorem \ref{thm:1} over $\{\bm{\beta}_i^{(t)}\}_{i=1}^{N_A}, \{\bm{\alpha}_k^{(t)}\}_{k=1}^{N_D}, \bm{\Sigma}_V^{(t)}$ and $\bm{c}^{(t)}$ at each round $t$. Since the optimization problem takes the same form over rounds, we simplify the notation by omitting the index $t$, and formulate the design problem as follows:
\begin{subequations}\label{opti:p1}
\begin{align}
 \nonumber \underset{\bm{\Sigma_V},\bm{c},\{\bm{\beta}_i\}_{i=1}^{N_A},\{\bm{\alpha}_k\}_{k=1}^{N_D}}{\textup{minimize}}&\quad \sum_{i_1\in[N_A]}\sum_{i_2\in[N_A]}[\bm{\Sigma_S}]_{i_1,i_2}\sum_{k\in\mathcal{N}_{D,i_1}}{\left(\sqrt{v_{k}}-c_{i_1}{\bm{\alpha}_{k}}^\mathsf{T}{\bm{H}_{i_1k}}^\mathsf{T}{\bm{\beta}_{i_1}}^\mathsf{\dag}\right)}^\mathsf{H}\\
 \nonumber&\times\sum_{k\in\mathcal{N}_{D,i_2}}\left(\sqrt{v_{k}}-c_{i_2}{\bm{\alpha}_{k}}^\mathsf{T}{\bm{H}_{i_2k}}^\mathsf{T}{\bm{\beta}_{i_2}}^\mathsf{\dag}\right)+\sum_{i\in[N_A]} \varepsilon_ic_{i}^2\|\bm{\beta}_i\|^2\\
 & +  \boldsymbol{c}^\mathsf{T} \boldsymbol{\Sigma}_S \boldsymbol{c}-\boldsymbol{c}^\mathsf{T} \boldsymbol{\Sigma}_S\left(\boldsymbol{\Sigma}_S+\boldsymbol{\Sigma}_V\right)^{-1} \boldsymbol{\Sigma}_S^\mathsf{T} \boldsymbol{c}, \label{p1_objec}\\
	\textup{subject to}\quad
	&\|\bm{\alpha}_k\|^2\leq \frac{P_k}{2}, \forall k\in[N_D],\label{p1_con1}\\
	& \frac{1}{2} \log \left(\frac{\operatorname{det}\left(\bm{\Sigma}_{S}+\bm{\Sigma}_{V}\right)}{\operatorname{det}\left(\bm{\Sigma}_{V}\right)}\right) \leq \sum_{i \in[N_A]} r_{i},\label{p1_con2}\\
	&\frac{1}{2} \log \left(\frac{\operatorname{det}\left(\bm{\Sigma}_S+\bm{\Sigma}_V\right)}{\operatorname{det}\left([\bm{\Sigma}_S+\bm{\Sigma}_V]_{\mathcal{K}^c}\right) \operatorname{det}\left([\bm{\Sigma}_V]_{\mathcal{K}}\right)}\right) \leq \sum_{i\in\mathcal{K}}r_i,\forall \mathcal{K} \subset[N_A], \mathcal{K} \neq \varnothing, \label{p1_con3}
\end{align}
\end{subequations}
where the power constraint in (\ref{p1_con1}) is from (\ref{device:power constraint}), the constraints (\ref{p1_con2})-(\ref{p1_con3}) are from Lemma \ref{RDin} and (\ref{bit con}). To solve the optimization problem in (\ref{opti:p1}), we propose an alternating optimization (AO) based algorithm, which is discussed in the following subsections.

%We then assume that $\rho= \mathbb{E}\left[\bm{r}_{k_1}^\mathsf{H}\bm{r}_{k_2}\right] /N$ for $\forall k_1, k_2 \in [N_D]$, and $\varepsilon = \mathbb{E}\left[\bm{Z}_i^\mathsf{H}\bm{Z}_i\right] /N$ for $\forall i \in [N_A]$ are known at the CS, to simplify problem in (\ref{opti:p1}) as follows
%\begin{subequations}\label{opti:p2}
	%\begin{align}
		%\min_{\bm{c},\bm{\Sigma}_V,\{\bm{\beta}_i\}_{i=1}^{N_A},\{\bm{\alpha}_k\}_{k=1}^{N_D}}\quad& \nonumber N_A\rho\sum_{i\in[N_A]}\left|\sum_{k\in\mathcal{N}_{D,i}}\left(\sqrt{v_{k}}-[\bm{c}]_i\bm{\alpha}_{k}^\mathsf{T}\bm{H}_{ik}^\mathsf{T}\bm{\beta}_{i}^\mathsf{\dag}\right)\right|^2+\varepsilon\sum_{i\in[N_A]}\|\bm{\beta}_i\|^2\\
		%\quad&\ + \bm{c}^\mathsf{T}\bm{\Sigma}_S\bm{c}-\bm{c}^\mathsf{T}\bm{\Sigma}_S(\bm{\Sigma}_S+\bm{\Sigma}_V)^{\mathsf{-1}}\bm{\Sigma}_S^\mathsf{T}\bm{c}  \\
		%\text{s.t.}\quad
		%&(\ref{p1_con1})-(\ref{p1_con3})
	%\end{align}
%\end{subequations}

\subsection{Optimization of each $\bm{\alpha}$ with fixed $\bm{\beta}, \bm{\Sigma}_V$ and $\bm{c}$}
With fixed $\{\bm{\alpha}_j\}_{j\ne k}$, $\{\bm{\beta}_i\}_{i=1}^{N_A},\bm{\Sigma_V}$ and $\bm{c}$, the problem in (\ref{opti:p1}) for optimizing $\bm{\alpha}_k$ reduces to
\begin{subequations}\label{problem r0}
	\begin{align}
		\nonumber\underset{\bm{\alpha}_k}{\textup{minimize}}\quad&  \sum_{i_1\in[N_A]}\sum_{i_2\in[N_A]}[\bm{\Sigma_S}]_{i_1,i_2}\sum_{k\in\mathcal{N}_{D,i_1}}{\left(\sqrt{v_{k}}-c_{i_1}{\bm{\alpha}_{k}}^\mathsf{T}{\bm{H}_{i_1k}}^\mathsf{T}{\bm{\beta}_{i_1}}^\mathsf{\dag}\right)}^\mathsf{H}\\
		&\times\sum_{k\in\mathcal{N}_{D,i_2}}\left(\sqrt{v_{k}}-c_{i_2}{\bm{\alpha}_{k}}^\mathsf{T}{\bm{H}_{i_2k}}^\mathsf{T}{\bm{\beta}_{i_2}}^\mathsf{\dag}\right),\label{problem3_obj}\\
		\text{subject to}\quad&\|\bm{\alpha}_k\|^2\leq \frac{P_k}{2},\label{problem3_con}
	\end{align}
\end{subequations}
where the objective function (\ref{problem3_obj}) and the constraint (\ref{problem3_con}) are both convex w.r.t. $\bm{\alpha}_k$. Therefore, (\ref{problem r0}) is a convex quadratically constrained quadratic programming (QCQP) problem that can be solved using standard convex optimization tools.

\subsection{Optimization of each $\bm{\beta}$ with fixed $\bm{\alpha}, \bm{\Sigma}_V$ and $\bm{c}$}
With fixed $\{\bm{\alpha}_k\}_{k=1}^{N_D}$, $\{\bm{\beta}_j\}_{j \neq i}, \bm{\Sigma}_V$ and $\bm{c}$, the problem in (\ref{opti:p1}) for optimizing $\bm{\beta}_i$ reduces to
\begin{align}\label{problem p4}
	\nonumber\underset{\bm{\beta}_i}{\textup{minimize}}\quad& \sum_{i_1\in[N_A]}\sum_{i_2\in[N_A]}[\bm{\Sigma_S}]_{i_1,i_2}\sum_{k\in\mathcal{N}_{D,i_1}}{\left(\sqrt{v_{k}}-c_{i_1}{\bm{\alpha}_{k}}^\mathsf{T}{\bm{H}_{i_1k}}^\mathsf{T}{\bm{\beta}_{i_1}}^\mathsf{\dag}\right)}^\mathsf{H}\\
	&\times\sum_{k\in\mathcal{N}_{D,i_2}}\left(\sqrt{v_{k}}-c_{i_2}{\bm{\alpha}_{k}}^\mathsf{T}{\bm{H}_{i_2k}}^\mathsf{T}{\bm{\beta}_{i_2}}^\mathsf{\dag}\right)+\varepsilon_ic_{i}^2\|\bm{\beta}_i\|^2.
\end{align}
Since the problem in (\ref{problem p4}) is a quadratic convex problem, we set the derivative of (\ref{problem p4}) with respect to $\bm{\beta}_i$ to zero, yielding a closed-form expression of $\bm{\beta}_i^\dagger$ as
\begin{equation}
	\bm{\beta}_i^\dagger=\left(\varepsilon_ic_i\bm{I}\!+\!\sum_{j\in[N_A]}\![\bm{\Sigma}_S]_{i,j}\!\sum_{k\in\mathcal{N}_{D,i}}\!\bm{H}_{ik}^\dagger\bm{\alpha}_k^\dagger\!\sum_{k\in\mathcal{N}_{D,j}}\!\bm{\alpha}_k^\mathsf{T}\bm{H}_{jk}^\mathsf{T}\right)^{-1}\!\sum_{j\in[N_A]}\![\bm{\Sigma}_S]_{i,j}\!\sum_{k\in\mathcal{N}_{D,j}}\!\sqrt{v_k}\!\sum_{k\in\mathcal{N}_{D,i}}\!\bm{H}_{ik}^\dagger\bm{\alpha}_k^\dagger,
\end{equation}
which is used for updating $\bm{\beta}_i$.

%\begin{equation}
%	\varepsilon_i\bm{\beta}_i+\bar{\rho}_i\sum_{k\in\mathcal{N}_{D,i}}[\bm{c}]_i\bm{H}_{ik}\bm{\alpha}_k \sum_{k\in\mathcal{N}_{D,i}}[\bm{c}]_i\bm{\alpha}_k^\mathsf{H}\bm{H}_{ik}^\mathsf{H}\bm{\beta}_i=\bar{\rho}_i\sum_{k\in\mathcal{N}_{D,i}}[\bm{c}]_i\bm{H}_{ik}\bm{\alpha}_k\sum_{k\in\mathcal{N}_{D,i}}\sqrt{v_k},
%\end{equation}
%which is simplified as
%\begin{equation}\label{beta_solve}
%	\bm{\beta}_i =\left(\varepsilon_i[\bm{c}]_{i}^2\bm{I}+\bar{\rho}_i[\bm{c}]_i^2\sum_{k\in\mathcal{N}_{D,i}}\bm{H}_{ik}\bm{\alpha}_k \sum_{k\in\mathcal{N}_{D,i}}\bm{\alpha}_k^\mathsf{H}\bm{H}_{ik}^\mathsf{H}\right)^{-1} \left(\bar{\rho}_i[\bm{c}]_i\sum_{k\in\mathcal{N}_{D,i}}\bm{H}_{ik}\bm{\alpha}_k\sum_{k\in\mathcal{N}_{D,i}}\sqrt{v_k}\right).
%\end{equation}

\subsection{Optimization of $\bm{c}$ with fixed $\bm{\alpha}, \bm{\Sigma}_V$ and $\bm{\beta}$}
With fixed $\{\bm{\alpha}_k\}_{k=1}^{N_D}, \{\bm{\beta}_i\}_{i=1}^{N_A}, \bm{\Sigma}_V$, the problem in (\ref{opti:p1}) for optimizing $\bm{c}$ is converted into 
	\begin{align}\label{problem r1}
		\nonumber\underset{\bm{c}}{\textup{minimize}}\quad&   \sum_{i_1\in[N_A]}\sum_{i_2\in[N_A]}[\bm{\Sigma_S}]_{i_1,i_2}\sum_{k\in\mathcal{N}_{D,i_1}}{\left(\sqrt{v_{k}}-c_{i_1}{\bm{\alpha}_{k}}^\mathsf{T}{\bm{H}_{i_1k}}^\mathsf{T}{\bm{\beta}_{i_1}}^\mathsf{\dag}\right)}^\mathsf{H}\\
		\nonumber&\times\sum_{k\in\mathcal{N}_{D,i_2}}\left(\sqrt{v_{k}}-c_{i_2}{\bm{\alpha}_{k}}^\mathsf{T}{\bm{H}_{i_2k}}^\mathsf{T}{\bm{\beta}_{i_2}}^\mathsf{\dag}\right)+\sum_{i\in[N_A]} \varepsilon_ic_{i}^2\|\bm{\beta}_i\|^2\\
		&+ \boldsymbol{c}^\mathsf{T}[ \boldsymbol{\Sigma}_S - \boldsymbol{\Sigma}_S\left(\boldsymbol{\Sigma}_S+\boldsymbol{\Sigma}_V\right)^{-1} \boldsymbol{\Sigma}_S^\mathsf{T}] \boldsymbol{c},
	\end{align}
where $\bm{\Sigma}_S-\bm{\Sigma}_S(\bm{\Sigma}_S+\bm{\Sigma}_V)^{\mathsf{-1}}\bm{\Sigma}_S^\mathsf{T}$ is positive semidefinite. As the problem in (\ref{problem r1}) is a quadratic convex problem, we set the derivative of (\ref{problem r1}) with respect to $\bm{c}$ to zero, yielding
%\begin{align}\label{optioptc}
%	\bm{c}=\left(\bm{T}+\bm{\Omega}+\bm{\Sigma}_S-\bm{\Sigma}_S(\bm{\Sigma}_S+\bm{\Sigma}_V)^{-1}\bm{\Sigma}_S^{\mathsf{T}}\right)^{-1}\!\begin{bmatrix}\begin{smallmatrix}
%			\sum_{j\in[N_A]}\sum{k\in\mathcal{N}_{D,j}}[\bm{\Sigma}_S]_{1,j}\sqrt{v_k}
%			\operatorname{Re}\{G_1\}\\\vdots\\\sum_{j\in[N_A]}\sum{k\in\mathcal{N}_{D,j}}[\bm{\Sigma}_S]_{N_A,j}\sqrt{v_k}
%			\operatorname{Re}\{G_{N_A}\}\end{smallmatrix}
%	\end{bmatrix},
%\end{align}
\begin{align}\label{optioptc}
	\bm{c}=\left(\bm{T}+\bm{\Omega}+\bm{\Sigma}_S-\bm{\Sigma}_S(\bm{\Sigma}_S+\bm{\Sigma}_V)^{-1}\bm{\Sigma}_S^{\mathsf{T}}\right)^{-1}\!\bm{\Sigma}_S\begin{bmatrix}
		\operatorname{Re}\{G_1\}\sum_{k\in\mathcal{N}_{D,1}}\sqrt{v_k}\\
		\vdots\\
		\operatorname{Re}\{G_{N_A}\}\sum_{k\in\mathcal{N}_{D,N_A}}\sqrt{v_k}
\end{bmatrix}
\end{align}
where $\bm{\Omega}\triangleq\operatorname{diag}(\varepsilon_1\|\bm{\beta}_1\|^2,\dots,\varepsilon_{N_A}\|\bm{\beta}_{N_A}\|^2)\in\mathbb{R}^{N_A\times N_A}$ is a dialog matrix, $\bm{T}\in{R}^{N_A\times N_A}$ is a matrix with each entry $[\bm{T}]_{i,j}=[\bm{\Sigma}]_{i,j}\operatorname{Re}\{G_i^\dagger G_j\}$, and $G_i\triangleq\sum_{k\in\mathcal{N}_{D,i}}\bm{\alpha}_k^{\mathsf{T}}\bm{H}_{ik}^{\mathsf{T}}\bm{\beta}_i^{\dagger}\in\mathbb{C}$.

\subsection{Optimization of $\bm{\Sigma}_V$ with fixed $\bm{\alpha}, \bm{c}$ and $\bm{\beta}$}
With fixed $\{\bm{\alpha}_k\}_{k=1}^{N_D}, \{\bm{\beta}_i\}_{i=1}^{N_A}, \bm{c}$, the problem in (\ref{opti:p1}) for optimizing $\bm{\Sigma}_V$ is reduced to 
\begin{subequations}\label{prob:sigma_v}
	\begin{align}
		\underset{\bm{\Sigma}_V}{\textup{maximize}}\quad&  \bm{c}^\mathsf{T}\bm{\Sigma}_S(\bm{\Sigma}_S+\bm{\Sigma}_V)^{\mathsf{-1}}\bm{\Sigma}_S^\mathsf{T}\bm{c}\label{prob:sigma_v,obj}  \\
		\text{subject to}\quad
		& \frac{1}{2} \log \left(\frac{\operatorname{det}\left(\bm{\Sigma}_{S}+\bm{\Sigma}_{V}\right)}{\operatorname{det}\left(\bm{\Sigma}_{V}\right)}\right) \leq \sum_{i \in[N_A]} r_{i}\label{optivcon1}\\
		&\frac{1}{2} \log \left(\frac{\operatorname{det}\left(\bm{\Sigma}_S+\bm{\Sigma}_V\right)}{\operatorname{det}\left([\bm{\Sigma}_S+\bm{\Sigma}_V]_{\mathcal{K}^c}\right) \operatorname{det}\left([\bm{\Sigma}_V]_{\mathcal{K}}\right)}\right) \leq \sum_{i\in\mathcal{K}}r_i,\forall \mathcal{K} \subset[N_A], \mathcal{K} \neq \varnothing.\label{optivcon2}
	\end{align}
\end{subequations}

We solve the optimization problem in (\ref{prob:sigma_v}) through an iterative algorithm based on majorization-minimization (MM). The algorithm starts with a feasible point called the \textit{current-point}. Each iteration round consists of two steps. In the first step, we construct a \textit{surrogate} problem, whose objective serves as a lower bound of the original objective with equality holds at the current-point. Besides, the feasible region of the surrogate problem should be a subset of the original feasible region and contains the current-point. In the second step, we solve the surrogate problem, and the solution will be used as the current-point in the next iteration. Following \cite{yang_federated_2022}, given a feasible point $\bm{\Sigma}_V=\widehat{\bm{\Sigma}}_V$, we provide a convex problem as the surrogate problem below
\begin{subequations}\label{prob:surrogate}
	\begin{align}
		\underset{\{\boldsymbol{\Sigma}_V\}}{\textup{maximize}}\quad & 2 \mathbf{c}^\mathsf{T} \boldsymbol{\Sigma}_S \mathbf{b}-\mathbf{b}^\mathsf{T}\left(\boldsymbol{\Sigma}_S+\boldsymbol{\Sigma}_V\right) \mathbf{b} \\
		\text { s.t. } \quad& \chi_{\mathcal{K}}\left(\mathbf{E}_{\mathcal{K}}, \mathbf{F}_{\mathcal{K}}, \boldsymbol{\Sigma}_V\right)+\xi_{\mathcal{K}}\left(\mathbf{E}_{\mathcal{K}}, \mathbf{F}_{\mathcal{K}}\right) \leq \sum_{i \in \mathcal{K}} r_i, \forall \mathcal{K} \subset[N_A], \mathcal{K} \neq \varnothing, \\
		& \chi_{[N_A]}\left(\bm{G}, \boldsymbol{\Sigma}_V\right)+\xi_{[N_A]}(\bm{G}) \leq \sum_{i \in[N_A]} r_i,
	\end{align}
\end{subequations}
where
\begin{subequations}
	\begin{align}
		\chi_{\mathcal{K}}\left(\mathbf{E}_{\mathcal{K}}, \mathbf{F}_{\mathcal{K}}, \boldsymbol{\Sigma}_V\right)=&\ \frac{\log (e)}{2} \!\operatorname{tr}\!\left\{\mathbf{F}_{\mathcal{K}}^{-1} [\boldsymbol{\Sigma}_V]_{\mathcal{K}}\right\}\!+\!\frac{\log (e)}{2}\! \operatorname{tr}\left\{\!\mathbf{E}^\mathsf{T}_{\mathcal{K}} \mathbf{F}_{\mathcal{K}}^{-1} \mathbf{E}_{\mathcal{K}} [\boldsymbol{\Sigma}_V]_{\mathcal{K}^c}\!\right\}\!-\!\frac{1}{2}\! \log \left(\operatorname{det}\left([\boldsymbol{\Sigma}_V]_{\mathcal{K}}\right)\right),
	\end{align}
	\begin{align}	
		\nonumber\xi_{\mathcal{K}}(\mathbf{E}_{\mathcal{K}}, \mathbf{F}_{\mathcal{K}})=&\  \frac{\log (e)}{2} \operatorname{tr}\left\{\mathbf{F}_{\mathcal{K}}^{-1}\left([\boldsymbol{\Sigma}_S]_{\mathcal{K}}+\mathbf{E}_{\mathcal{K}} [\boldsymbol{\Sigma}_S]_{\mathcal{K}^c} \mathbf{E}^\mathsf{T}_{\mathcal{K}}-\mathbf{E}_{\mathcal{K}}[ \boldsymbol{\Sigma}_S]_{\mathcal{K}^c, \mathcal{K}}-[\boldsymbol{\Sigma}_S]_{\mathcal{K}, \mathcal{K}^c} \mathbf{E}^\mathsf{T}_\mathcal{K}\right)\right\}\\
		&+\frac{1}{2} \log (\operatorname{det}(\mathbf{F}_{\mathcal{K}}))-\frac{|\mathcal{K}| \log (e)}{2},
	\end{align}
	\begin{equation}
		\chi_{[N_A]}\left(\mathbf{G}, \boldsymbol{\Sigma}_V\right)=\frac{\log (e)}{2} \operatorname{tr}\left\{\mathbf{G}^{-1} \boldsymbol{\Sigma}_V\right\}-\frac{1}{2} \log \left(\operatorname{det}\left(\boldsymbol{\Sigma}_V\right)\right), 
	\end{equation}
	\begin{equation}
		\xi_{[N_A]}(\mathbf{G})=\frac{1}{2} \log (\operatorname{det}(\mathbf{G}))+\frac{\log (e)}{2} \operatorname{tr}\left\{\mathbf{G}^{-1} \boldsymbol{\Sigma}_X\right\}-\frac{N_A \log (e)}{2},
	\end{equation}
\end{subequations}
and $\mathbf{b}=(\boldsymbol{\Sigma}_S+\widehat{\bm{\Sigma}}_V)^{-1} \boldsymbol{\Sigma}_S \mathbf{c}$, $\bm{G}=\bm{\Sigma}_S+\bm{\widehat{\Sigma}}_V$, $\mathbf{E}_{\mathcal{K}}=[\bm{\Sigma_S}]_{\mathcal{K}, \mathcal{K}^c}([\bm{\Sigma}_S+\widehat{\bm{\Sigma}}_V]_{\mathcal{K}^c})^{-1}$, and $\mathbf{F}_{\mathcal{K}}=[\bm{\Sigma}_S+\bm{\widehat{\Sigma}}_V]_{\mathcal{K}}-[\bm{\Sigma_S}]_{\mathcal{K}, \mathcal{K}^c}([\bm{\Sigma}_S+\widehat{\bm{\Sigma}}_V]_{\mathcal{K}^c})^{-1} [\bm{\Sigma_S}]_{\mathcal{K}^c, \mathcal{K}}$ for all nonempty set $\mathcal{K}\subset{[N_A]}$.  Since the surrogate problem (\ref{prob:surrogate}) is convex and is solved optimally with existing convex optimization solvers such as CVXPY \cite{diamond_cvxpy_2016}. By repeatedly constructing and solving this surrogate problem following the MM framework introduced before, we finally obtain a suboptimal solution to the original problem (\ref{prob:sigma_v}). The overall optimization process is outlined in Algorithm \ref{algo:optimization}. 

\begin{algorithm}
 	\caption{The Overall Optimization Process.}
 		\begin{algorithmic}[1]\label{algo:optimization} %这个1 表示每一行都显示数字
 			\STATE Each device $k$ computes $v_k, P_k$ and sends them to AP $i$, if $k\in\mathcal{N}_{D,i}$
 			
 			\STATE Each AP $i$ receives $\{v_k,  P_k\}_{k\in\mathcal{N}_{D,i}}$ from the devices and sends them to the CS
 			\STATE Each AP $i$ estimates $\{\bm{H}_{ik}\}_{k\in\mathcal{N}_{D,i}}$, $\varepsilon_i$ and sends them to the CS
 			\STATE The CS receives $\{\{\bm{H}_{ik}\}_{k\in\mathcal{N}_{D,i}}, \varepsilon_i,\{v_k,  P_k\}_{k\in\mathcal{N}_{D,i}}\}_{i=1}^{N_A}$ from the APs
 			\STATE The CS estimates $\bm{\Sigma}_S$
 			
 			\REPEAT
 			\STATE The CS optimizes $\{\bm{\alpha}_k\}_{k=1}^{N_D}$ via solving the problem in (\ref{problem r0})
 			
 			\STATE The CS optimizes $\{\bm{\beta}_i\}_{i=1}^{N_A}$ via (\ref{problem p4})

 			\STATE The CS optimizes $\bm{\Sigma}_V$ via iteratively solving the surrogate problem in (\ref{prob:sigma_v})

 			\STATE The CS optimizes $\bm{c}$ via (\ref{problem r1})
 			\UNTIL {convergence}
 			\STATE The CS sends $\{\bm{\alpha}_k\}_{k\in\mathcal{N}_{D,i}}, \bm{\beta}_i$ to each AP $i$
 			\STATE Each AP $i$ sends $\bm{\alpha}_k$ to the device $k$, if $k\in\mathcal{N}_{D,i}$
 		\end{algorithmic} 	
 \end{algorithm}

\section{Practical Design}\label{coder}
\begin{figure}
	\centering
	\includegraphics[width=0.95\linewidth]{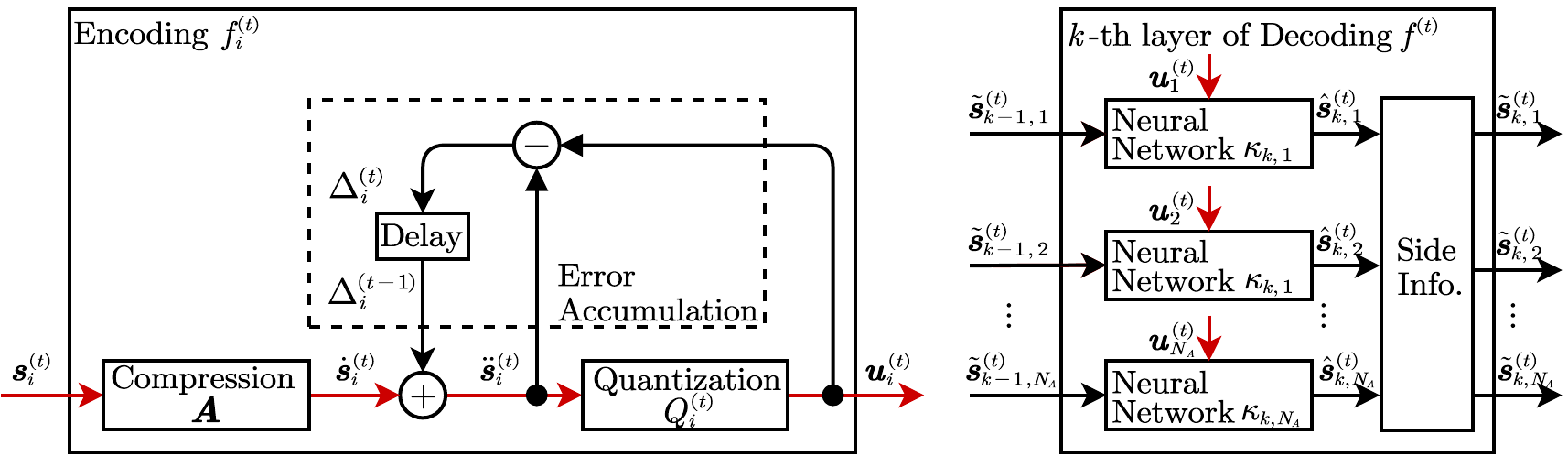}
	\caption{An illustration of the encoding function $f_i^{(t)},\forall i$ and layer $k$ on decoding function $f^{(t)}$.}
	\label{fig:coder}
\end{figure}
In this section, we provide a practical design of encoding functions $\{f_i^{(t)}\}_{i=1}^{N_A}$ and joint decoding function $f^{(t)}$ . In this design, we consider a quadratic distortion measure $d(\mathbf{a}^{(t)}, \mathbf{b}^{(t)})=\|\mathbf{a}^{(t)}-\mathbf{b}^{(t)}\|^2$ and the linear aggregation target $\kappa\left(\bm{a}_1^{(t)}, \ldots, \bm{a}_{N_A}^{(t)}\right)=\sum_{i=1}^{N_A} c_i^{(t)} \bm{a}_i^{(t)}$.
\subsection{Encoding Function Design}
 As shown in Fig. \ref{fig:coder}, we propose an encoding function $f_i^{(t)},\forall i$ that comprises a compression module, quantization mudule and error accumulation module. Specifically, at the $t$-th round, the source vector $\bm{s}_i^{(t)}$ is firstly compressed into a low-dimensional vector as
\begin{equation}\label{Nencoder_co}
	\dot{\bm{s}}_i^{(t)} = \bm{As}_i^{(t)}\in\mathbb{R}^{N\sigma}, \forall i\in[N_A],
\end{equation}
where $\bm{A}\in\mathbb{R}^{N\sigma\times N}$ is a random compression matrix with a compression ratio of $\sigma\in(0,1]$. To reduce quantization error from the previous round, $\dot{\bm{s}}_i^{(t)}$ is added with an error accumulation term \cite{seide_1-bit_2014} as
\begin{equation}\label{Nencoder_ac}
	\ddot{\bm{s}}_i^{(t)} = \dot{\bm{s}}_i^{(t)} + \bm{\Delta}_i^{(t-1)} \in\mathbb{R}^{N\sigma}, \forall i\in[N_A],
\end{equation}
where the error accumulation term $\bm{\Delta}_i^{(t-1)}\in\mathbb{R}^{N\sigma}$ is from the round $t-1$. To meet  the rate-constraint in Definition \ref{defi1} and (\ref{bit con}), i.e., $\frac{1}{N} \log \left(B_i^{(t)}\right) \leq r_i^{(t)}+\epsilon, \forall i \in[N_A]$, the resulting vector $\ddot{\bm{s}}_i^{(t)}$ is next quantized to
\begin{equation}\label{Nencoder_en}
	\bm{u}_i^{(t)} = \mathcal{Q}_i^{(t)}(\ddot{\bm{s}}_i^{(t)}, r_i^{(t)})\in[B_i^{(t)}], \forall i\in[N_A],
\end{equation}
where $B_i^{(t)}$ given in (\ref{encoder}) is the codebook size of encoding function $f_i^{(t)}$, $\mathcal{Q}_i^{(t)}$ is an A-law quantizer that discretizes each element of $\ddot{\bm{s}}_i^{(t)}$ into a quantized number, resulting in vector $\bm{u}_i^{(t)}$. The error accumulation vector $\bm{\Delta}_i^{(t)}$ is calculated as
\begin{equation}\label{Nencoder_de}
	\bm{\Delta}_i^{(t)} = \ddot{\bm{s}}_i^{(t)}-\bm{u}_i^{(t)}\in\mathbb{R}^{N\sigma}, \forall i\in[N_A].
\end{equation}
Finally, each AP $i$ sends the quantized vector $\bm{u}_i^{(t)}$ to the CS through the uplink fronthaul network.

\begin{algorithm}
	\caption{Practical Design of Encoding and Decoding Functions.}
	\begin{algorithmic}[1]\label{algo4} %这个1 表示每一行都显示数字

		\STATE \textbf{Each encoding function} $f_i^{(t)}$ \textbf{in} (\ref{encoder}) \textbf{does:}

		\STATE  Compute $\bm{u}_i^{(t)}$ via (\ref{Nencoder_co})-(\ref{Nencoder_en})
		\STATE  Compute $\bm{\Delta}_i^{(t)}$ via (\ref{Nencoder_de})
		\algrule
		\STATE \textbf{The decoding function} $f^{(t)}$ \textbf{in} (\ref{decoder}) \textbf{does:}
		\STATE  Initialize $\tilde{\bm{s}}_{0,i}^{(t)}=\bm{0},\forall i$
		\STATE  \textbf{for} {$k=1,\dots,K:$}	
		\STATE   \quad Compute $\{\hat{\bm{s}}_{k,i}^{(t)}, \tilde{\bm{s}}_{k,i}^{(t)}\}_{i=1}^{N_A}$ via (\ref{decoder:straight})-(\ref{decoder:side})
		\STATE   \textbf{end}
		\STATE   Compute $\hat{\bm{z}}^{(t)}$ via (\ref{z hat})
		\algrule
		\STATE \textbf{The training process w.r.t. $\bm{w}^{(t)}$:}
		\STATE  Initialize $\tilde{\bm{s}}_{0,i}^{(t)}=\bm{0},\forall i,t$
		
		\STATE \textbf{repeat}

		\STATE \quad Sample data $\bm{s}_1^{(t)},\dots,\bm{s}_{N_A}^{(t)}$ from a memoryless Gaussian source $\mathcal{N}(\bm{0},\bm{\Sigma}_S^{(t)})$

		\STATE \quad Compute the loss in (\ref{decoder train})
		\STATE \quad Update $\bm{w}^{(t)}$ through backpropagation
		\STATE  \textbf{until} convergence
	\end{algorithmic} 	
	
\end{algorithm}

\subsection{Decoding Function Design}
As shown in Fig. \ref{fig:coder}, we propose a decoding function $f^{(t)}$ with $K$ layers, each consisting of $N_A$ neural networks and a side information module. At layer $k$, each neural network $\kappa_{k,i}^{(t)}(\cdot), \forall i$ computes
\begin{equation}\label{decoder:straight}
	\hat{\bm{s}}_{k,i}^{(t)} = \kappa_{k,i}^{(t)}(\bm{u}_i^{(t)}, \tilde{\bm{s}}_{k-1,i}^{(t)})\in\mathbb{R}^{N}, \forall i\in[N_A],
\end{equation}
where $\bm{u}_i^{(t)}$ is the encoded vector from AP $i$, and $\tilde{\bm{s}}_{i,k-1}^{(t)}\in\mathbb{R}^{N},\forall i$ is a side information module from layer $k-1$, with $\tilde{\bm{s}}_{i,0}^{(t)}$ initialized to $\bm{0}$. Then the side information module at layer $k$ collects $\{\hat{\bm{s}}_{k,i}^{(t)}\}_{i=1}^{N_A}$ to generate the side information vector as
\begin{equation}\label{decoder:side}
	\tilde{\bm{s}}_{k,i}^{(t)}=\sum_{j\in[N_A],j\ne i}\hat{\bm{s}}_{k,j}^{(t)}\in\mathbb{R}^{N}, \forall i\in[N_A].
\end{equation}
Considered a linear aggregation target, i.e., $\kappa\left(\bm{a}_1^{(t)}, \ldots, \bm{a}_{N_A}^{(t)}\right)=\sum_{i=1}^{N_A} c_i^{(t)} \bm{a}_i^{(t)}$, the encoding function $f^{(t)}$ constructs the estimation $\hat{\bm{z}}^{(t)}$ as
\begin{equation}\label{z hat}
	\hat{\bm{z}}^{(t)}=\sum_{i=1}^{N_A} c_i^{(t)}\bm{s}_{K,i}^{(t)},
\end{equation}
where $\bm{s}_{K,i}^{(t)}$ is the vector from  layer $K$. To train the neural networks in $f^{(t)}$ at each round $t$, based on the quadratic distortion measure and the linear aggregation target, we design an empirical objective as
\begin{equation}\label{decoder train}
	\underset{\bm{w}^{(t)}}{\textup{maximize}}\left\|\bm{z}^{(t)}-\hat{\bm{z}}^{(t)}\right\|^2\triangleq\left\|\sum_{i\in[N_A]}c_i^{(t)}\bm{s}_i^{(t)}-f^{(t)}(\bm{u}_1^{(t)},\dots,\bm{u}_{N_A}^{(t)},\bm{w}^{(t)})\right\|^2,
\end{equation}
where $\bm{u}_i^{(t)}=f_i^{(t)}(\bm{s}_i^{(t)})$, $\bm{w}\in\mathbb{R}^M$ is the network parameter of decoding $f^{(t)}$ with dimension $M$, and $s_{1,n}^{(t)}$, $\dots$, $s_{N_A,n}^{(t)}, \forall n$ are sampled from a memoryless Gaussian source $\mathcal{N}(\bm{0},\bm{\Sigma}_S^{(t)})$. The minimization of (\ref{decoder train}) is through the gradient descent (GD) w.r.t. $\bm{w}$. The overall deign of encoding functions and decoding function is summarized in Algorithm \ref{algo4}.

\subsection{Optimization in Practical Design}
In practical encoding and decoding function designs, $\bm{\Sigma}_V^{(t)}$ is not an optimization variable. Therefore, the optimization of practical design is similar to the problem in (\ref{opti:p1}) but does treat $\bm{\Sigma}_V^{(t)}$ as a known constant. To estimate $\bm{\Sigma}_V^{(t)}$, we follow \cite{yang_federated_2022}, to set the compression ratio $\sigma=1$ for each AP and define the $i$-th diagonal element of $\bm{\Sigma}_V^{(t)}$ as $\frac{1}{N}(\bm{s}_i^{(t)}-\bm{u}_i^{(t)})^\mathsf{T}(\bm{s}_j^{(t)}-\bm{u}_j^{(t)})$. In practice, we require each AP to send a sub-vector of $\bm{s}_i-\bm{u}_i^{(t)}$ to the CS intermittently for estimation during training.

\section{Numerical Results}\label{results}

\subsection{Experimental Settings}
We validate our proposed OA-FL in MIMO Cloud-RAN framework with experiments, and provide some schemes for comparison:
\begin{itemize}
	\item Error-free bound: This bound assumes that the CS receives all the local updates of devices in an error-free fashion and updates the global model by (\ref{model updating0}).
	
	\item L-DSC bound: In our proposed OA-FL in MIMO Cloud-RAN framework, this bound considers L-DSC encoding at each AP and L-DSC decoding at the CS \cite{yang_federated_2022}. We emphasize that the L-DSC encoding and decoding serve as a benchmark for theoretical performance limit.
	
	\item Proposed practical design: In our proposed OA-FL in MIMO Cloud-RAN framework, this scheme uses the practical design proposed in Section \ref{coder} as the encoding and decoding functions.

	\item Quantization: In our proposed OA-FL in MIMO Cloud-RAN framework, this scheme is applied as follows: each AP $i$ computes the encoded vector $\bm{u}_i^{(t)}$ as $\bm{u}_i^{(t)}=\mathcal{Q}_i^{(t)}(\bm{s}_i^{(t)},r_i^{(t)}),\forall t,i$, and the CS computes the global update $\hat{\bm{z}}^{(t)}$ as $\hat{\bm{z}}^{(t)}=\sum_{i=1}^{N_A}c_i^{(t)}\bm{u}_i^{(t)}\in \mathbb{R}^N,\forall t$.

	\item Distributed deterministic information bottleneck (DDIB): In our proposed OA-FL in MIMO Cloud-RAN framework, this scheme considers DDIB encoding at each AP and DDIB decoding at the CS, both of which are based on deep neural networks (DNN). 
\end{itemize}
To make performance comparisons, we conduct federated image classification experiments on three datasets: the MNIST dataset of handwritten digits \cite{yann_mnist_1998}, the Fashion-MNIST dataset of fashion clothing \cite{xiao_fashion-mnist_2017} and the CIFAR-10 dataset. We train a neural network on each device and the CS with two $5\times5$ convolution layers (the first with $10$ channels, the second with $20$, each followed by a $2\times2$ max pooling operation), a fully connected layer with 50 units and ReLU activation, and a final softmax output layer (model parameter length $N = 21840$). Each device performs 5 stochastic gradient descent updates with a learning rate of $0.01$ and a local batch size of 1200. The CS updates the global model using a learning rate of $1.5/(1+t/10)$, where $t$ is the communication round. The channel gain is modeled according to \cite{goldsmith_wireless_2005}, where the channel gain is given by $\mathbf{H}_{ik}^{(t)}=\sqrt{G_{\mathrm{R},i} G_{\mathrm{T},k} \nu \delta_{ik}^{-\alpha}} \tilde{\mathbf{H}}_{ik}^{(t)}$. Here, the entries of $\tilde{\mathbf{H}}_{ik}^{(t)}$ are modeled as i.i.d. circularly symmetric complex Gaussian random variables with zero-mean and unit-variance, $G_{\mathrm{R},i}$ and $G_{\mathrm{T},k}$ are the antenna gains at AP $i$ and device $k$, respectively, $\alpha$ is the path loss exponent, $\delta_{ik}$ is the distance between device $k$ and AP $i$, and $\nu$ is the path loss at a reference distance of 1 m \cite{wu_intelligent_2019}. We summarize the simulation settings in Tab. \ref{table:parameter}.

\begin{table}
	\centering
	\caption{Simulation Settings}
	\scriptsize
	\begin{tabular}{|l|l||l|l||l|l||l|l||l|l|}
		\hline Parameter & Value & Parameter & Value & Parameter & Value & Parameter & Value &Parameter & Value \\
		\hline$N_{\mathrm{T}}$ & 3 & $N_{\mathrm{R}}$ & 8 & $\alpha$ & 3.8 & $\nu$ & -60 dB&$\delta_{ik}$&30 m\\
		\hline$G_{\mathrm{R},i}$ & 10 dBi & $G_{\mathrm{T},k}$ & 5 dBi & $P_k$ & 1 W & $\sigma_z^2$ & -80 dBm&$\sigma$&0.5 \\
		\hline
	\end{tabular}
		\normalsize
	\label{table:parameter}
\end{table}

\subsection{Comparisons of the Proposed Algorithms Under Various Settings}
To verify the effectiveness of side information module in the proposed practical design, we present a comparison method called "Proposed practical design (zero side information)", where the side information vector $\tilde{\bm{s}}_{k,i}^{(t)},\forall k,i,t$ is set to $\bm{0}$. Fig. \ref{fig:data1} shows the performance of our proposed practical design in utilizing the inter-AP correlation. The left of Fig. \ref{fig:data1} shows the training loss of the proposed practical design, i.e., the objective in (\ref{decoder train}). The middle of Fig. \ref{fig:data1} shows the system mean square error (MSE), i.e. $\|\bm{e}\|^2$ versus round $t$. We observe that the proposed practical design achieves lower training loss and MSE than the proposed practical design (zero side information), thanks to the ability of the side information module to leverage inter-AP correlation. Moreover, as we increased the number of layers $K$, we see a decrease in both the training loss and system MSE. This improvement comes at the expense of increased computational complexity and latency, which is a reasonable trade-off. The right of Fig. \ref{fig:data1} shows the test accuracy on MNIST dataset versus round $t$. We see that by utilizing the side information module, the test accuracy of the proposed practical design was close to the L-DSC bound and the error-free bound, with a deviation of less than 1\% at convergence. Additionally, the proposed practical design significantly outperformed the proposed practical design (zero side information), which further demonstrates the advantage of leveraging the inter-AP correlation.
\begin{figure}
	\centering
	\includegraphics[width=0.8\linewidth]{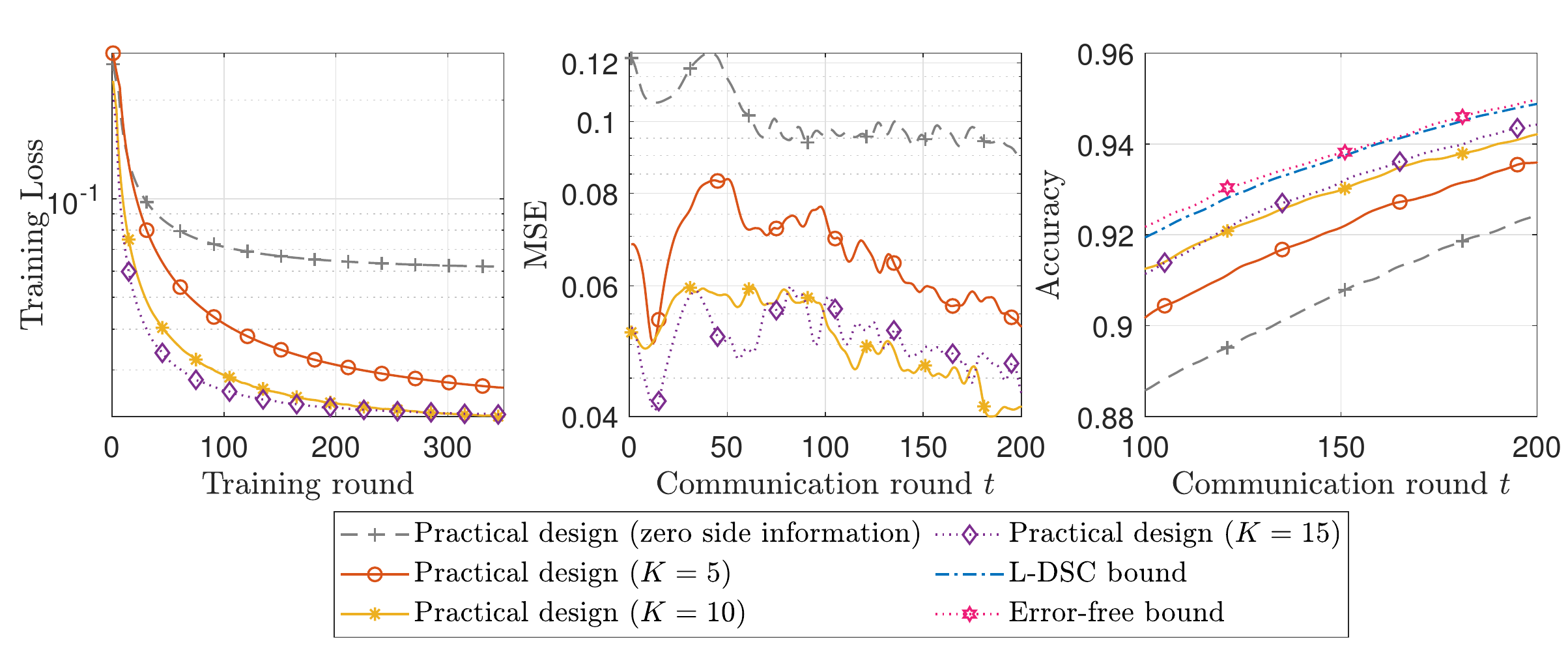}
	\caption{\textbf{Left}: Training loss of the proposed practical design, i.e., the objective in (\ref{decoder train}). \textbf{Middle}: The system mean square error on MNIST dataset, i.e. $\|\bm{e}\|^2$ versus round $t$. \textbf{Right}: The test accuracy on MNIST dataset versus round $t$. Other setting: $N_A=3, \mathcal{N}_{D,1}={1,\dots,5}, \mathcal{N}_{D,2}={6,\dots,13}, \mathcal{N}_{D,3}={14,\dots,20}$, $r_i^{(t)}=1 \text{ bit/symbol}, \forall i,t, K=10$.}
	\label{fig:data1}
\end{figure}

\begin{figure}
	\centering
	\includegraphics[width=0.8\linewidth]{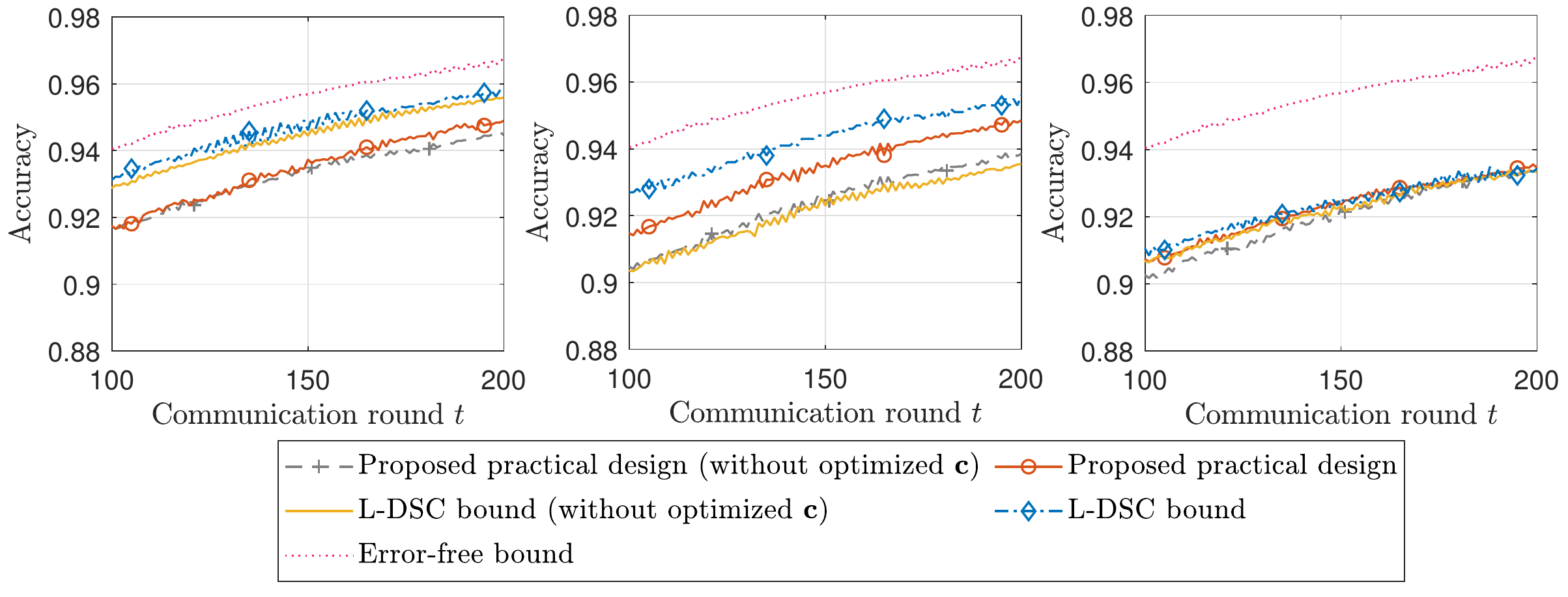}
	\caption{The curves show the test accuracy and mean-square error (MSE) on the MNIST dataset versus training round $t$. The settings are as follows: $N_A=3$, $\mathcal{N}_{D,1}={1,\dots,5}$, $\mathcal{N}_{D,2}={6,\dots,13}$, and $\mathcal{N}_{D,3}={14,\dots,20}$. \textbf{Left}: $P_k^{(t)}=1$W, $\forall k,t$, $r_i^{(t)}=10 \text{ bit/symbol},\forall i,t$. \textbf{Middle}: $P_k^{(t)}=0.5$W, $\forall k\in\mathcal{N}_{D,1}$, $P_k^{(t)}=1$W $\forall k\in\mathcal{N}_{D,2}\cup\mathcal{N}_{D,3}$. $r_1^{(t)}=2 \text{ bit/symbol}, r_2^{(t)}=r_3^{(t)}=10 \text{ bit/symbol}$. \textbf{Right}: $P_k^{(t)}=0.5$W, $\forall k,t$, $r_i^{(t)}=2 \text{ bit/symbol},\forall i,t, K=10$.}
	\label{fig:data2}
\end{figure}

To demonstrate the effectiveness of the optimization process in Algorithm \ref{algo:optimization}, we designed three scenarios, which are depicted in the left, middle, and right sections of Fig. \ref{fig:data2}. In the first scenario, all devices have the same transmission power, and all APs have the same information rate limitation, e.g., $P_k^{(t)}=1W, \forall k,t, \tilde{r}_i^{(t)}=10 \text{ bit/symbol}, \forall i,t$. In the second scenario, devices have varying transmission power and APs have differing information rate limitations, e.g., $P_k^{(t)}=0.5$W, $\forall k\in\mathcal{N}_{D,1}$, $P_k^{(t)}=1$W $\forall k\in\mathcal{N}_{D,2}\cup\mathcal{N}_{D,3}$. $r_1^{(t)}=2 \text{ bit/symbol}, r_2^{(t)}=r_3^{(t)}=10 \text{ bit/symbol}$. In the third scenario, we follow a similar settings as in the first scenario, but with reduced transmission power for all devices and lowered information rate limitation for all APs, e.g., $P_k^{(t)}=0.5$W, $\forall k,t$, $r_i^{(t)}=2 \text{ bit/symbol},\forall i,t$. We also included two comparison methods, "Proposed practical design (without optimized $\bm{c}$)" and "L-DSC bound (without optimized $\bm{c}$)". In these methods, the weighted parameter is fixed at $\bm{c}=\bm{1}$. Fig. \ref{fig:data2} shows that optimizing $\bm{c}$ results in a small improvement in performance in the left and right plots, but a significant improvement in the middle plot. Specifically, in the middle plot, the L-DSC bound achieves an accuracy improvement of approximately 0.017 at the convergence point, while the proposed practical design achieves an improvement of approximately 0.011. These results suggest that optimizing $\bm{c}$ enhance the learning performance when the communication quality across different APs is unbalanced, as in scenario (2). By examining the optimization problem of $\bm{c}$ in (\ref{problem r1}), we see that the optimization of $\bm{c}$ is actually a resource allocator that measures the wired communication quality of different APs to the CS and the wireless communication quality of different APs to their respective served devices. Additionally, we noted that in scenarios where the channel environment is degraded, i.e., the right plot, the test accuracy of the proposed practical design is comparable to the L-DSC bound.

\begin{figure}
	\centering
	\includegraphics[width=0.8\linewidth]{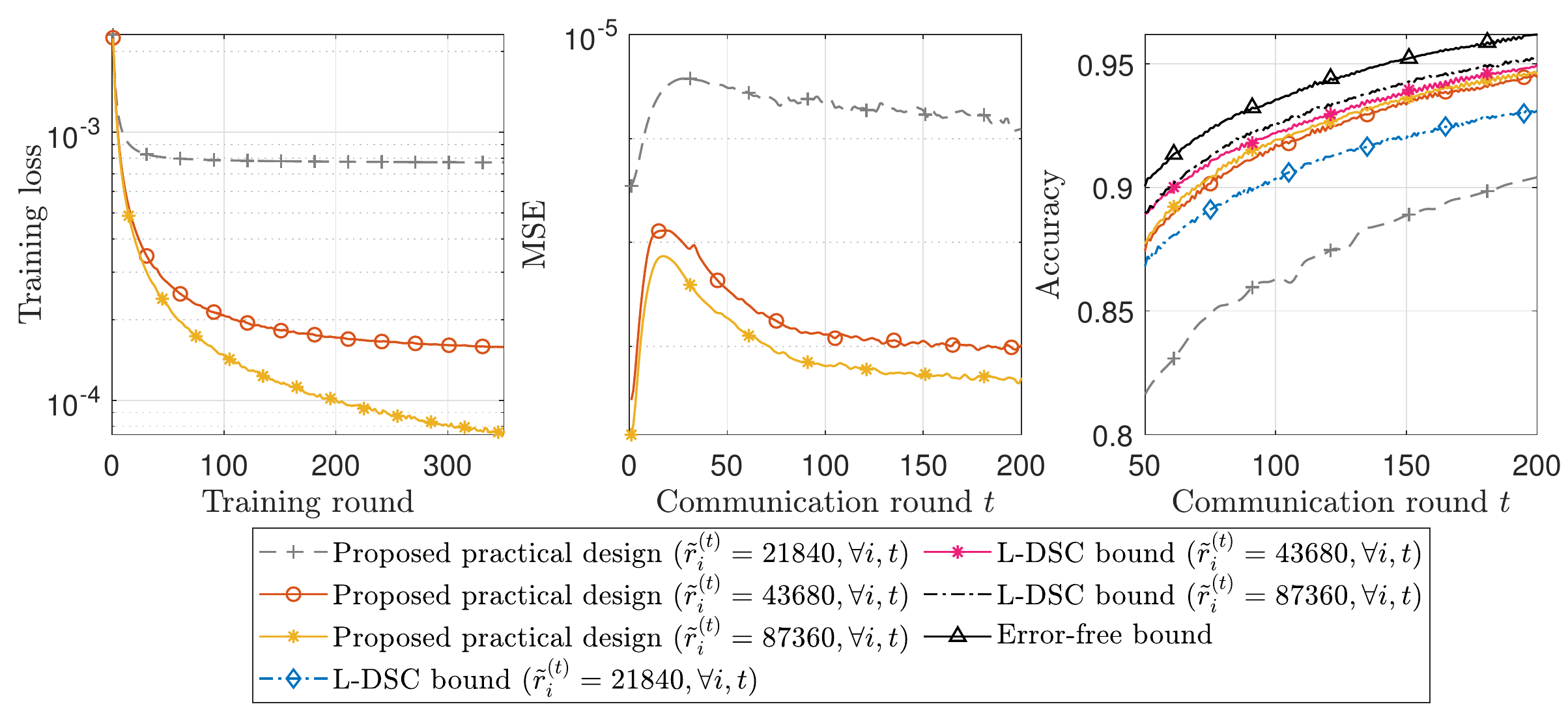}
	\caption{\textbf{Left}: Training loss of the proposed practical design, i.e., the objective in (\ref{decoder train}). \textbf{Middle}: The system mean square error on MNIST dataset, i.e. $\|\bm{e}\|^2$ versus training round $t$. \textbf{Right}: The test accuracy on MNIST dataset versus training round $t$. Other setting: $N_A=3,\mathcal{N}_{D,1}={1,\dots,5},\mathcal{N}_{D,2}={6,\dots,13},\mathcal{N}_{D,3}={14,\dots,20}, K=10$.}
	\label{fig:data3}
\end{figure}

To demonstrate the performance of the proposed practical design in various rate constraints, we provide Fig. \ref{fig:data3}. In this figure, we present three different rate constraint settings, as denoted by $\tilde{r}_i^{(t)}$, which determine the maximum number of bits that AP $i$ can transmit to the CS during each training round $t$. The left and middle of Fig. \ref{fig:data1} show the training loss and MSE of proposed practical design under various rate constraints. We see that the training loss and MSE of the proposed practical design exhibit a significant increase when the rate constraint, i.e., $\tilde{r}_i^{(t)}$, is close to the data dimension $N$, compared to when it is close to $2N$ or $4N$. The right of Fig. \ref{fig:data3} compares the test accuracy of our practical design with the L-DSC bound and the error-free bound. At convergence, the accuracy of our practical design is within 0.01 of the L-DSC bound when $\tilde{r}_i^{(t)}$ is equal to $4N$ or $2N$, and within 0.025 when $\tilde{r}_i^{(t)}$ is equal to $N$.

\subsection{Comparisons With Existing Schemes}

\begin{figure}
	\centering
	\includegraphics[width=0.8\linewidth]{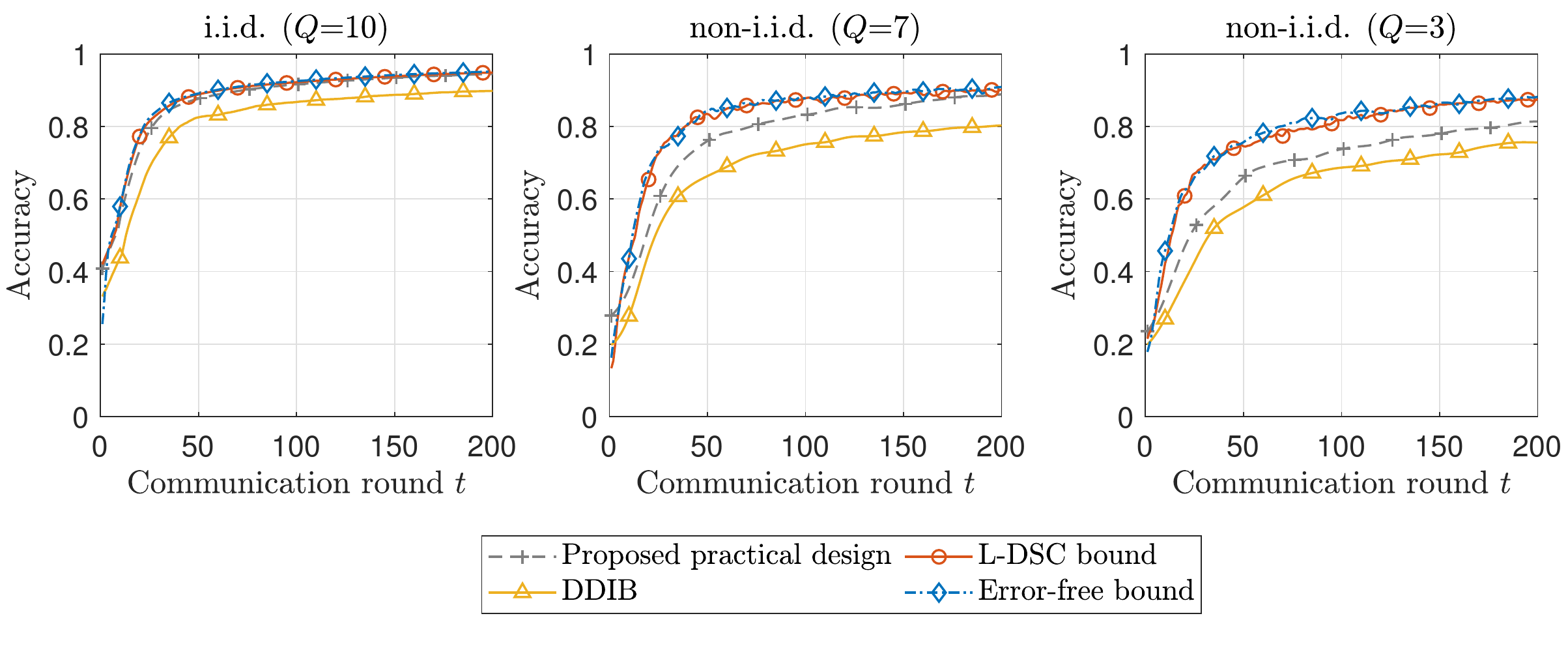}
	\caption{Test accuracy on the MNIST task, where $Q$ measures the heterogeneous setting to a certain extent. \textbf{Left}: Independent and identically distributed (i.i.d.) dataset with $Q=10$. \textbf{Middle}: non-i.i.d. dataset with $Q=7$. \textbf{Right}: non-i.i.d. dataset with $Q=3$. Other setting: $N_A=3,\mathcal{N}_{D,1}={1,\dots,5},\mathcal{N}_{D,2}={6,\dots,13},\mathcal{N}_{D,3}={14,\dots,20}, K=10$.}
	\label{fig:data6}
\end{figure}

We empirically validate the heterogeneous setting of federated datasets in our framework using the MNIST task. Specifically, each device randomly selects $Q$ classes from the dataset $\{\mathcal{B}_k\}_{k=1}^{N_D}$, and subsequently samples are randomly drawn without repetition from these selected classes. When $Q$ is set to 10, the local datasets of the devices are independently and identically distributed (i.i.d.). Fig. \ref{fig:data6} compares the test accuracy between proposed practical design and other baseline schemes under different heterogeneous settings. We observe that as the value of $Q$ decreases, the learning performance of these schemes also decreases. We also see that the proposed practical design achieves a learning performance comparable to the L-DSC bound at the convergence point when $Q$ is set to either 10 or 7. Furthermore, the proposed practical design outperforms the quantization scheme and the DDIB scheme across all heterogeneous settings.

\section{Conclusions}
In this paper, we proposed the OA-FL in MIMO Cloud-RAN framework to address the issues of limited server coverage and low resource utilization in traditional OA-FL. We also noted that inter-AP correlation can be leveraged to improve learning performance and therefore proposed modeling the global aggregation stage as an L-DSC problem to make analysis from the perspective of rate-distortion theory. We further analyzed the performance of the proposed OA-FL in MIMO Cloud-RAN framework. Based on this analysis, we formulated a communication-learning optimization problem to improve the system performance by considering the inter-AP correlation. We showed that, under mild conditions, the problem can be solved using alternating optimization (AO) and majorization-minimization (MM). Subsequently, we proposed a practical design that demonstrates the utilization of inter-AP correlation. The numerical results show that the proposed practical design effectively leverages inter-AP correlation, and outperforms other baseline schemes.

\appendices
\section{System Error Bound}\label{app0}
To start with, we bound $\mathbb{E}[||\bm{e}^{(t)}||^2]$ as
\begin{align}\label{expectation error}
	\mathbb{E}[||\bm{e}^{(t)}||^2] \leq 2(\mathbb{E}[||\bm{e}_{1}^{(t)}||^2]+\mathbb{E}[||\bm{e}_{2}^{(t)}||^2]),
\end{align}
by using the triangle inequality and the inequality of arithmetic means. With (\ref{device:complex}), (\ref{acc:receive})-(\ref{acc:process}), (\ref{error analysis}), we have
\begin{align}
	\mathbb{E}[\|\bm{e}_1^{(t)}\|^2]
	= \mathbb{E}\left[\left\|\sum_{i\in[N_A]}\sum_{k\in\mathcal{N}_{D,i}}\left(\sqrt{v_k^{(t)}}-c^{(t)}_i{\bm{\alpha}_k^{(t)}}^\mathsf{T}{\bm{H}_{ik}^{(t)}}^\mathsf{T}{\bm{\beta}_i^{(t)}}^\mathsf{\dag}\right)\bm{r}_k^{(t)}-\sum_{i\in[N_A]}c^{(t)}_i{\bm{Z}_{i}^{(t)}}^{\mathsf{T}}{\bm{\beta}_i^{(t)}}^\dagger\right\|^2\right].
\end{align}	
Since the term $\bm{Z}_i,\forall i\in[N_A]$ is an independent Gaussian random matrix, we have 
\begin{align}
	\mathbb{E}[\|\bm{e}_1^{(t)}\|^2]\leq&\  \mathbb{E}\left[\left\|\sum_{i\in[N_A]}\!\sum_{k\in\mathcal{N}_{D,i}}\!\left(\!\sqrt{v_k^{(t)}}\!-\!c^{(t)}_i\!{\bm{\alpha}_k^{(t)}}^\mathsf{T}\!{\bm{H}_{ik}^{(t)}}^\mathsf{T}\!{\bm{\beta}_i^{(t)}}^\mathsf{\dag}\right)\!\bm{r}_k^{(t)}\!\right\|^2\right]\!+\! \sum_{i\in[N_A]}\mathbb{E}\left[\left\|c^{(t)}_i{\bm{Z}_{i}^{(t)}}^{\mathsf{T}}{\bm{\beta}_i^{(t)}}^\dagger\right\|^2\right].
\end{align}	
From the assumption in Lemma \ref{all_bound}, we have
\begin{align}\label{error1}
	\mathbb{E}[||\boldsymbol{e}_1^{(t)}||^2]\leq &\  N\sum_{i_1\in[N_A]}\sum_{i_2\in[N_A]}[\bm{\Sigma_S}^{(t)}]_{i_1,i_2}\sum_{k\in\mathcal{N}_{D,i_1}}{\left(\sqrt{v_{k}^{(t)}}-c_{i_1}^{(t)}{\bm{\alpha}_{k}^{(t)}}^\mathsf{T}{\bm{H}_{i_1k}^{(t)}}^\mathsf{T}{\bm{\beta}_{i_1}^{(t)}}^\mathsf{\dag}\right)}^\mathsf{H}\\
	\nonumber&\times\sum_{k\in\mathcal{N}_{D,i_2}}\left(\sqrt{v_{k}^{(t)}}-c_{i_2}^{(t)}{\bm{\alpha}_{k}^{(t)}}^\mathsf{T}{\bm{H}_{i_2k}^{(t)}}^\mathsf{T}{\bm{\beta}_{i_2}^{(t)}}^\mathsf{\dag}\right) + N\sum_{i\in[N_A]}\varepsilon_i^{(t)} {c^{(t)}_i}^2\|\bm{\beta}_i^{(t)}\|^2,
\end{align}
where $\varepsilon_i\triangleq\frac{1}{N}\mathbb{E}[\|\bm{Z}_i^{(t)}\|^2]\in\mathbb{R}, \forall i \in [N_A]$. The closed-form expressions of the terms $\mathbb{E}[\|\bm{e}_1^{(t)}\|^2]=\frac{1}{N}\mathbb{E}[\|\bm{z}^{(t)}-\hat{\bm{z}}^{(t)}\|^2]$ is given in Lemma \ref{RDin}. From (\ref{expectation error}), (\ref{error1}) and Lemma \ref{RDin}, we have the closed-form bound of $\mathbb{E}[||\bm{e}^{(t)}||^2]$ in (\ref{error_term}), which completes the proof.

%\begin{align}
%	\mathbb{E}[\|\bm{e}_1^{(t)}\|^2]
%	= \mathbb{E}\left[\left\|\sum_{i\in[N_A]}\!\sum_{k\in\mathcal{N}_{D,i}}\!\left(\sqrt{v_k^{(t)}}\!-\![\bm{c}^{(t)}]_i{\bm{\alpha}_k^{(t)}}^\mathsf{T}\!{\bm{H}_{ik}^{(t)}}^\mathsf{T}{\bm{\beta}_i^{(t)}}^\mathsf{\dag}\right)\bm{r}_k^{(t)}\right\|^2-\sum_{i\in[N_A]}[\bm{c}]_i\bm{Z}_{i}^{\mathsf{T}\bm{\beta}_i\right]\!+\!\sum_{i\in[N_A]}\!\mathbb{E}\left[\left\|{\bm{Z}_i^{(t)}}^\mathsf{T}{\bm{\beta}_i^{(t)}}^\dag\right\|^2\right]\label{app:error1_bound_s2}
%	\end{align}	

	%\nonumber=&\ \sum_{i_1\in[N_A]}\sum_{i_2\in[N_A]}\sum_{k_1\in\mathcal{N}_{D,i_1}}\sum_{k_2\in\mathcal{N}_{D,i_2}}\left(\sqrt{v_{k_1}^{(t)}}-[\bm{c}^{(t)}]_{i_1}{\bm{\alpha}_{k_1}^{(t)}}^\mathsf{T}{\bm{H}_{i_1k_1}^{(t)}}^\mathsf{T}{\bm{\beta}_{i_1}^{(t)}}^\mathsf{\dag}\right)^\mathsf{H}\\
	%&\times\left(\sqrt{v_{k_2}^{(t)}}-[\bm{c}^{(t)}]_{i_2}{\bm{\alpha}_{k_2}^{(t)}}^\mathsf{T}{\bm{H}_{i_2k_2}^{(t)}}^\mathsf{T}{\bm{\beta}_{i_2}^{(t)}}^\mathsf{\dag}\right)\mathbb{E}\left[{\bm{r}_{k_1}^{(t)}}^\mathsf{H}\bm{r}_{k_2}^{(t)}\right] + N\sum_{i\in[N_A]} \varepsilon_i^{(t)}\|{\bm{\beta}_i^{(t)}}\|^2 \\

\bibliographystyle{IEEEtran}
\bibliography{output}

\end{document}